\documentclass[12pt,letterpaper]{article}

\usepackage[utf8]{inputenc}
\usepackage[T1]{fontenc}
\usepackage[left=1.25in,right=1.25in,top=1.25in,bottom=1.25in]{geometry}
\usepackage{amsmath,amssymb,amsfonts,amsthm}
\usepackage{graphicx}
\usepackage{textcomp}
\usepackage[authoryear]{natbib}
\usepackage{enumitem}
\usepackage{float}

\renewcommand{\baselinestretch}{1.30345}
\renewcommand{\footnotesize}{\tiny}
\makeatletter
\long\def\@footnotetext#1{\insert\footins{%
    \reset@font\footnotesize
    \interlinepenalty\interfootnotelinepenalty
    \splittopskip\footnotesep
    \splitmaxdepth \dp\strutbox \floatingpenalty \@MM
    \hsize\columnwidth \@parboxrestore\raggedright
    \protected@edef\@currentlabel{%
       \csname p@footnote\endcsname\@thefnmark
    }%
    \color@begingroup
      \@makefntext{%
        \rule\z@\footnotesep\ignorespaces#1\@finalstrut\strutbox}%
    \color@endgroup}}
\renewcommand\section{\@startsection {section}{1}{\z@}%
  {-2.8ex \@plus -1ex \@minus -.2ex}{1.6ex \@plus.2ex}%
  {\normalfont\Large\bfseries\raggedright\hyphenpenalty=\@M}}
\renewcommand\subsection{\@startsection{subsection}{2}{\z@}%
  {-2.5ex\@plus -1ex \@minus -.2ex}{1.0ex \@plus .2ex}%
  {\normalfont\large\bfseries\raggedright\hyphenpenalty=\@M}}
\renewcommand\subsubsection{\@startsection{subsubsection}{3}{\z@}%
  {-2.5ex\@plus -1ex \@minus -.2ex}{1.0ex \@plus .2ex}%
  {\normalfont\normalsize\bfseries\raggedright\hyphenpenalty=\@M}}
\def\@listi{\leftmargin\leftmargini
  \parsep 2\p@ \@plus1\p@ \@minus1\p@
  \topsep 4\p@ \@plus2\p@ \@minus2\p@
  \itemsep 2\p@ \@plus1\p@ \@minus1\p@}
\let\@listI\@listi
\@listi
\def\thm@space@setup{%
  \thm@preskip=5\p@ \@plus2\p@ \@minus2\p@
  \thm@postskip=\thm@preskip}
\g@addto@macro\normalsize{%
  \abovedisplayskip .05in \@plus2\p@ \@minus2\p@
  \belowdisplayskip .05in \@plus2\p@ \@minus2\p@
  \abovedisplayshortskip \z@ \@plus2\p@
  \belowdisplayshortskip .05in \@plus2\p@ \@minus2\p@}
\makeatother

\theoremstyle{plain}
\newtheorem{theorem}{Theorem}
\newtheorem{corollary}{Corollary}
\newtheorem{proposition}{Proposition}
\newtheorem{lemma}{Lemma}

\theoremstyle{definition}
\newtheorem{definition}{Definition}
\newtheorem{remark}{Remark}
\newcommand{\E}{\mathbb E}
\newcommand{\Prob}{\Pr}

\begin{document}

\title{Outcome Disclosure and Temporal Refinement in Multi-Battle Team Contests}
\author{Bo Chen\thanks{Bo Chen: Shenzhen University, email: bochenbonn@gmail.com.}
\and Rui Gao\thanks{Rui Gao: Shandong Youth University of Political Science, email: gaorui@sdyu.edu.cn.}
\and Jingfeng Lu\thanks{Jingfeng Lu: National University of Singapore, email: ecsljf@nus.edu.sg.}
\and Zhewei Wang\thanks{Zhewei Wang: Shandong University, email: zheweiwang@sdu.edu.cn.}}
\date{August 2026}
\maketitle

\begin{abstract}
\baselineskip=14.5pt
A team-contest designer values output rather than expenditure;
nonlinear conversion makes the distinction consequential. We
study two
non-pecuniary instruments in majority-rule
contests decided by pairwise all-pay battles with private abilities:
disclosing resolved outcomes and splitting the battle schedule into
finer blocks.
Neither changes a battle's average pivotality; each only redistributes
it across histories. Under nested information
structures, this redistribution makes equilibrium ability-scaled
expenditure weakly more dispersed player by player without changing its
mean; expected aggregate expenditure is invariant across all designs
considered.
The resulting convex-order comparison ranks expected total output: convex
output costs favor no disclosure and coarser temporal structures, concave
costs favor full disclosure and finer ones, and linear costs make both
comparisons neutral. The rankings hold for finite-support and smooth
continuous-type ability distributions. The mechanism extends to
degree-zero component contest technologies with a unique equilibrium
outcome distribution. No
and full disclosure bound every admissible public garbling.

\bigskip

\thispagestyle{empty}
\end{abstract}

\newpage
\setcounter{page}{1}

\section{Introduction}\label{sec:intro}

Many team contests are decided not by a single aggregate action but by a
collection of component battles. Political parties compete district by
district, military coalitions fight across multiple fronts, and litigants
and sports teams meet in sequences of pairwise matchups. We study
environments in which every scheduled battle is completed and the overall
winner is determined by majority rule, as in an election in which every
district votes, a league in which every fixture is played, or a portfolio
of projects that are all carried out. The value of winning any particular
battle then depends on whether that battle can affect the aggregate
outcome.

What a designer of such a contest values is what contestants produce---%
votes mobilized, cases won, projects delivered---rather than the
resources they burn producing it. This distinction is especially
consequential when resources convert into performance nonlinearly. A
benchmark result establishes a strong neutrality: \citet{fu-lu-pan-2015}
show that in
majority-rule team contests neither total expected effort nor the overall
contest outcome depends on the temporal structure or feedback policy. We
show that the corresponding neutrality in our environment is a statement
about resources. Expected aggregate expenditure is invariant to
prior-outcome disclosure and to the temporal structure for every conversion
technology; expected total output is not.

The two design instruments operate through the same mechanism. Neither
changes how likely a battle is to be pivotal on average, since that is
fixed by the primitives; each can only redistribute pivotality across
histories.
Without prior-outcome disclosure, players evaluate a battle by its ex
ante likelihood of affecting the team prize; with disclosure, they
condition on outcomes already resolved, so the same battle is highly
consequential after some histories and almost irrelevant after others.
Redistribution of this kind is irrelevant for anything linear and
decisive for anything curved, and the curvature that decides it is that
of the technology converting expenditure into output.\footnote{The role
of cost curvature is also related to \citet{moldovanu-sela-2001}, who
show that effort-cost curvature shapes optimal prize allocation. Our
design instruments are different: disclosure and temporal structure
redistribute incentive intensity across histories rather than prizes
across ranks.}

We model each component battle as an all-pay auction with privately known
abilities, in which contestants choose costly expenditures that a
possibly nonlinear technology converts into output. Two component-battle
properties drive every design result. First, the probability that a team
wins a battle does not depend on the level of its effective prize spread,
so disclosed histories leave the distribution of battle outcomes
untouched. Second, equilibrium ability-scaled expenditure is proportional
to that spread. This scaling makes the expected-output response inherit
the curvature of the technology: convex output costs generate a concave
response, concave costs a convex response, and linear costs a linear
response. Both properties hold for finite-support and for smooth
continuous-type ability distributions. They also arise under a broader
class of component contests whose success functions are homogeneous of
degree zero in the linearly costly action, provided the component
equilibrium outcome distribution is unique.

Our central result, Theorem~\ref{thm:information-refinement}, delivers a
distributional comparison stronger than neutrality of expected
expenditure. Whenever one admissible information structure refines
another, equilibrium ability-scaled expenditure under the finer structure
is a mean-preserving spread of its counterpart under the coarser
structure, player by player. Equality of expected aggregate
expenditure alone would not rank nonlinear output: total output is the
sum of the contestants' transformed expenditures, not the transformation
of their aggregate expenditure. The player-level convex-order relation is
what permits Jensen's inequality to rank output.

For a fixed temporal structure, full prior-outcome disclosure induces a
mean-preserving spread of each battle's effective prize spread relative
to no disclosure. Hence no disclosure generates weakly higher expected
total output under convex output costs, whereas full
disclosure generates weakly higher output under concave costs; the two
policies are output-equivalent under linear costs. The same ranking
extends to committed, public, and common garblings of previously
resolved outcomes, with no and full disclosure providing the two bounds.

Temporal refinement operates through the same informational channel.
Under full disclosure, splitting a block creates additional disclosure
dates and makes the effective prize spreads of affected battles weakly
more dispersed without changing their conditional means. Consequently,
temporal refinement weakly reduces expected total output under
convex costs, weakly increases it under concave costs, and is neutral
under linear costs. Under no prior-outcome disclosure, the temporal
structure is irrelevant because previous outcomes are never observed.
These conclusions require no symmetry across component battles and rank
temporal structures only by refinement.

Expenditure neutrality, by contrast, is exact. Equilibrium expenditure in
a component battle is proportional to its effective prize spread, and a
linear function of that spread is insensitive to any redistribution
preserving its mean. Expected aggregate expenditure is therefore the same
under no disclosure, under full disclosure, under every garbling in the
class above, and under every ordered temporal structure, for every
admissible technology; we give its common value in closed form. The design
problem is thus about the distribution and conversion of expenditure into
output, not about the average resources contestants are induced to spend.

\citet{fu-lu-pan-2015} introduce the majority-rule team contest with
pairwise battles studied here, and the design instruments examined since
have been prizes and player assignment. \citet{feng-lu-2018} and
\citet{feng-jiao-kuang-lu-2024} characterize optimal prize structures,
while \citet{fu-lu-2020} and \citet{konishi-pan-simeonov-2022} study the
ordering and matching of players; \citet{konrad-kovenock-2009} and
\citet{hafner-2017} analyze related multi-battle formats. We hold prizes,
abilities, and the matching fixed and vary only what contestants know
about pivotality, an instrument that is non-pecuniary and leaves the
component battles themselves unchanged.
\citet{barbieri-serena-2024} also compare temporal structures, in a
contest-success-function model with the winner's output as the
objective. A broader literature on dynamic and sequential competition
studies how the history of play shapes incentives through momentum,
discouragement, and last-stand effects, including
\citet{harris-vickers-1987}, \citet{strumpf-2002},
\citet{klumpp-polborn-2006}, \citet{gelder-2014},
\citet{barbieri-serena-2022}, \citet{fan-kuang-lu-2025}, and
\citet{moller-beccuti-2025}. \citet{malueg-yates-2010} test contest
theory in best-of-three tennis matches, a format with early stopping
that our completion requirement excludes.

In the literature on disclosure in contests and all-pay auctions, the
disclosed object is typically an ability or an aggregate of abilities,
so that disclosure operates by changing beliefs about the strength of
the opposition; see \citet{fu-jiao-lu-2014}, \citet{zhang-zhou-2016},
\citet{chen-jiang-knyazev-2017}, \citet{lu-ma-wang-2018},
\citet{chen-ma-zhu-zhou-2020}, \citet{fu-wu-2022},
\citet{serena-2022}, and \citet{chen-serena-2023}. Here the disclosed
object is a battle outcome, and the channel is pivotality: disclosure
leaves beliefs about abilities and the component-game primitives
unchanged but reallocates the value of winning a battle across
histories. At the component level, our model is an all-pay contest with
privately known abilities as in \citet{konrad-kovenock-2010}. The closest
antecedent is \citet{chen-lu-bai-2026}, whose regularization of
two-player all-pay auctions with discrete private values is the
construction we use to characterize each battle. We extend the resulting
analysis to nonlinear output and develop a quantile-matching analogue for
smooth continuous types.

Finally, we extend the analysis beyond the finite-support,
own-type-observability baseline. The disclosure and temporal-refinement
results continue to hold under smooth continuous types and, subject to
uniqueness of the equilibrium outcome distribution, under general
degree-zero component contest technologies. Under strictly positive battle-specific prizes,
the fixed-structure disclosure ranking also extends to local type
observability and global type revelation.

The rest of the paper is organized as follows.
Section~\ref{sec:model} presents the model, and
Section~\ref{sec:example} illustrates the main mechanism.
Section~\ref{sec:equilibrium} characterizes equilibrium, first in a
component battle and then in the overall contest.
Section~\ref{sec:information-refinement} establishes the general
information-refinement theorem. Sections~\ref{sec:fixed-disclosure}
and~\ref{sec:temporal-refinement} apply it to outcome disclosure and
temporal refinement, respectively, and Section~\ref{sec:expenditure}
establishes expenditure neutrality across both nested and nonnested
designs. Section~\ref{sec:extensions} presents the extensions, and
Section~\ref{sec:conclusion} concludes.

\section{Model}\label{sec:model}

This section defines the multi-battle team contest, the temporal structure, and the finite-support private-information baseline. Subsection~\ref{subsec:continuous-types} replaces that baseline with a smooth continuous-type specification, under which all of the design results continue to hold.

\subsection{Multi-battle team contest}\label{subsec:multi-battle}

Two teams, indexed by $i=A,B$, compete for a team prize that each member
of the winning team values at $V>0$. There are $2n+1$ pairwise component
battles, where $n\geq1$ is an integer, indexed by
\[
\mathcal T=\{1,2,\ldots,2n+1\}.
\]
Each team consists of $2n+1$ players, and each player takes part in
exactly one battle. The two players assigned to battle $t$ compete
head-to-head; we write $i(t)$ for team $i$'s player in that battle and
$-i(t)$ for his opponent. The matching is exogenous and fixed, so
neither team chooses which of its players appears in which battle. The
order in which battles are played is chosen by the designer, as
described in Subsection~\ref{subsec:blocks}. A team wins the team prize
if and only if it wins at least $n+1$ of the $2n+1$ battles. All
scheduled battles are carried out, even after one team has secured the
majority.

Let $W_t\in\{A,B\}$ denote the winner of battle $t$ and let
$\mathbf 1\{\cdot\}$ denote the indicator function, so that team $A$
wins the team prize if and only if
\[
\sum_{t\in\mathcal T}\mathbf 1\{W_t=A\}\geq n+1,
\]
and otherwise team $B$ wins. Each battle also awards its winner a
battle-specific prize $\pi_t\geq0$, valued equally by the two assigned
players. Player $i(t)$'s payoff is therefore
\begin{equation}
\Pi^{i(t)}
=
V\mathbf 1\{\text{team }i\text{ wins the team prize}\}
+
\pi_t\mathbf 1\{W_t=i\}
-
e^{i(t)},
\label{eq:payoff}
\end{equation}
where $e^{i(t)}\geq0$ is the expenditure he chooses in battle $t$.

Battle $t$ is pivotal for the team prize if and only if the other $2n$
battles are split evenly, that is, if and only if team $A$ wins exactly
$n$ of them. Define
\begin{equation}
D_t=
\left\{
\sum_{s\neq t}\mathbf 1\{W_s=A\}=n
\right\}.
\label{eq:pivotal-event}
\end{equation}
When $D_t$ occurs, the winner of battle $t$ secures the team prize for
his team; when $D_t$ does not occur, the outcome of battle $t$ does not
affect the allocation of the team prize. Hence, by~\eqref{eq:payoff} and
holding the outcomes of the other $2n$ battles fixed, the difference
between a contestant's payoff from winning battle $t$ and his payoff
from losing it, at any given expenditure, is
\begin{equation}
Y_t=\pi_t+V\mathbf 1\{D_t\}.
\label{eq:Yt}
\end{equation}
This difference is the same for the two players assigned to battle $t$.

When choosing expenditures, those two players evaluate the battle using
its \emph{effective prize spread}, the expectation of $Y_t$ conditional
on the public information available when the battle is played, that is,
$\pi_t$ plus $V$ times the conditional probability that battle $t$ is
pivotal. Because $Y_t$ depends only on the outcomes of other battles,
independence across battles implies that conditioning additionally on
either current contestant's type does not change this expectation. The
effective prize spread is therefore common to the two players; in the
component-battle analysis below, $v$ denotes a generic value of it.

\subsection{Temporal structure and disclosure}\label{subsec:blocks}

The contest is organized according to a temporal structure, which is an ordered partition
\[
\mathcal P=(\mathcal B_1,\ldots,\mathcal B_L)
\]
of the battle set $\mathcal T$ into nonempty blocks, where $\mathcal B_\ell$ denotes the $\ell$th block, for $\ell=1,\ldots,L$, with $\mathcal B_\ell\cap\mathcal B_{\ell'}=\emptyset$ for all $\ell\neq\ell'$ and $\bigcup_{\ell=1}^{L}\mathcal B_\ell=\mathcal T$. The two extremes are
the simultaneous structure $\mathcal P^{\mathrm{sim}}=(\mathcal T)$,
which consists of a single block, and the singleton-block structures, in
which every block contains one battle. We write $\ell(t)$ for the index
of the block containing battle $t$.

Battles are carried out sequentially from block $1$ to block $L$. Battles in the same block are carried out simultaneously. Hence, players
in the current block do not observe the outcomes of battles in that block
before choosing their expenditure levels. Only outcomes resolved in
preceding blocks are available for disclosure when the current block
begins; whether these outcomes are revealed depends on the disclosure
policy. Because all battles are carried out, the fact that a later block
is reached does not itself reveal earlier outcomes. This completion
requirement excludes standard early-stopping formats, in which the
cancellation of later battles would itself reveal information about
previous outcomes. The requirement therefore delimits the design space
rather than simplifying the analysis: under early stopping, the fact
that a battle is played reveals that neither team has yet secured the
majority, so the no-disclosure policy cannot be implemented and the two
policies below are not comparable.

The designer commits ex ante to one of two outcome-disclosure policies,
which Subsection~\ref{subsec:disclosure-ranking} later embeds in a
broader class of committed public garblings.
\begin{itemize}[leftmargin=*]
    \item \textbf{No prior-outcome disclosure $(N)$}: Before each block $\ell$, players observe none of the battle outcomes resolved in preceding blocks.

    \item \textbf{Full prior-outcome disclosure $(F)$}: Before each block $\ell$, players observe all battle outcomes resolved in preceding blocks.
\end{itemize}
The temporal structure and the disclosure policy are publicly announced before play begins. Except for the outcomes revealed by the policy, players receive no public information about earlier battles; in particular, type realizations, expenditures, and outputs in other battles are not disclosed. The designer chooses only the ordered temporal structure and the committed public outcome-disclosure policy. The team prize $V$, battle-specific prizes $\{\pi_t\}$, contestant assignments, ability distributions, technology $\gamma$, and the requirement that every battle be completed are exogenous.

\subsection{Types, technology, and payoffs}\label{subsec:types-output}

In battle $t$, player $i(t)$'s ability, or type, is denoted by
$\alpha^{i(t)}$, and types are drawn independently across players and
battles. In the finite-support baseline, player $i(t)$'s ability
distribution has finite support
\[
\left\{\alpha_{1}^{i(t)},\ldots,\alpha_{J_t^{i}}^{i(t)}\right\},
\qquad
\alpha_1^{i(t)}>\cdots>\alpha_{J_t^{i}}^{i(t)}>0,
\]
where $J_t^i$ is the number of primitive ability types of player $i(t)$
in battle $t$, with probabilities
\[
\left(p_1^{i(t)},\ldots,p_{J_t^{i}}^{i(t)}\right),
\qquad
\sum_{k=1}^{J_t^{i}}p_k^{i(t)}=1.
\]
Each player observes only his own type before choosing his expenditure,
and contestants cannot communicate privately. The ability distributions
and all other model primitives are common knowledge.

Each component battle is an all-pay auction. Player $i(t)$ chooses an
irreversible expenditure $e^{i(t)}\geq0$, which generates the
(effective) output
\[
x^{i(t)}
=
\gamma\left(\alpha^{i(t)}e^{i(t)}\right),
\]
where $\gamma:\mathbb R_+\to\mathbb R_+$ is continuous, twice
continuously differentiable on $\mathbb R_{++}$, and satisfies
$\gamma(0)=0$ and $\gamma'(y)>0$ for $y>0$; in particular, $\gamma$ is
strictly increasing. Equivalently, the cost of any feasible output
$x^{i(t)}\in\gamma(\mathbb R_+)$ is
\[
e^{i(t)}
=
\frac{\gamma^{-1}\left(x^{i(t)}\right)}
{\alpha^{i(t)}}.
\]
Throughout the paper, total output refers to the sum of the outputs $x$,
and aggregate expenditure to the sum of the costly inputs $e$. Players
are risk neutral, and the designer's objective is to maximize expected
total output; contestants' expected aggregate expenditure is treated
separately in Section~\ref{sec:expenditure}.

The player who generates the higher output wins the battle. When the
effective prize spread is positive, ties are resolved by fair lotteries
that are independent across battles and independent of all type draws;
outcomes of battles whose effective prize spread is zero are recorded
according to the convention introduced in
Subsection~\ref{subsec:overall-contest}. The main comparisons focus
on the cases in which $\gamma$ is globally concave, globally convex, or
linear. Because the cost of output is $\gamma^{-1}(x)/\alpha$, these
cases correspond respectively to convex, concave, and linear
\emph{output costs}, and we use that terminology throughout.

By~\eqref{eq:Yt}, a contestant's expected gain from winning battle $t$
rather than losing it is its effective prize spread $v$. Subtracting the
action-independent expected payoff from losing, player $i(t)$'s
normalized payoff from an output profile is therefore
\begin{equation}
u^{i(t)}\!\left(x^{i(t)},x^{-i(t)}\right)
=
\begin{cases}
v-e^{i(t)}, & \text{if }x^{i(t)}>x^{-i(t)},\\
-e^{i(t)}, & \text{if }x^{i(t)}<x^{-i(t)},\\
\frac{v}{2}-e^{i(t)}, & \text{if }x^{i(t)}=x^{-i(t)},
\end{cases}
\label{eq:component-payoff}
\end{equation}
where $e^{i(t)}=\gamma^{-1}\left(x^{i(t)}\right)/\alpha^{i(t)}$ is the
expenditure that generates output $x^{i(t)}$.

\section{Illustrative Example}\label{sec:example}

This section illustrates how information about pivotality affects
expected output in the simplest nontrivial setting.
\medskip

\textbf{Setup.}\quad
Three battles decide a contest between teams $A$ and $B$, with a team
winning if it takes at least two. Normalize the team prize to one and set every
battle-specific prize to zero, so that a battle's effective prize spread
is simply the conditional probability that it is pivotal. Abilities are
deterministic, a one-type special case of the finite-support
environment, with pairs
\[
(\alpha^{A(1)},\alpha^{B(1)})=(2,1),\qquad
(\alpha^{A(2)},\alpha^{B(2)})=(1,1),\qquad
(\alpha^{A(3)},\alpha^{B(3)})=(1,2),
\]
so team $A$ is stronger in battle $1$, the players are evenly matched in
battle $2$, and team $A$ is weaker in battle $3$. The technology is
$\gamma(y)=y^\rho$ with $\rho>0$, so that the output cost of a player of
ability $\alpha$ is $x^{1/\rho}/\alpha$, and $\rho<1$, $\rho=1$, and
$\rho>1$ correspond to convex, linear, and concave output costs.
\medskip

\textbf{Component battles.}\quad
With deterministic abilities each battle is a complete-information
all-pay auction. Multiplying a player's
payoff~\eqref{eq:component-payoff} by his ability $\alpha>0$ is a
positive rescaling, so it leaves best responses unchanged while turning
the battle into an all-pay auction with the common cost function
$\gamma^{-1}$ and valuations $\alpha^{A(t)}v$ and $\alpha^{B(t)}v$:
ability enters only through the induced valuation, and the battle is an
instance of the general-cost all-pay auction of
\citet{kaplan-wettstein-2006}. For abilities $a\geq b$, their
Corollary~1 gives the stronger player's winning probability as
$1-b/(2a)$, independent of the positive level of $v$, and integrating
$x=\gamma(y)$ against the equilibrium distributions of their Theorem~1
gives expected total output
\[
\phi(v;a,b)
=
\frac{1}{v}
\left(\frac{1}{a}+\frac{1}{b}\right)
\int_0^{bv}\gamma(y)\,dy.
\]
The first three columns of Table~\ref{tab:example-battles} record the
resulting winning probability $\mu_t$ of team $A$ and output response
$\phi_t$ in each battle, the latter inheriting the curvature of
$\gamma$.
\medskip

\textbf{Disclosure and refinement.}\quad
Compare three designs: the simultaneous structure
$\mathcal P^{\mathrm{sim}}=(\{1,2,3\})$, under which no outcome precedes
any battle and disclosure is vacuous; the singleton-block structure
$\mathcal P^{\mathrm{seq}}=(\{1\},\{2\},\{3\})$ with full prior-outcome
disclosure; and the intermediate $\mathcal P^{12,3}=(\{1,2\},\{3\})$,
which resolves battles $1$ and $2$ together and discloses their outcomes
before battle $3$. A battle is pivotal exactly when the other two are
split, so under $\mathcal P^{\mathrm{sim}}$ battle $1$ faces the spread
$\mu_2(1-\mu_3)+(1-\mu_2)\mu_3=1/2$; the remaining entries of
Table~\ref{tab:example-battles} follow in the same way.

\begin{table}[H]
\centering
\setlength{\tabcolsep}{5pt}
\renewcommand{\arraystretch}{1.15}
\caption{Component battles and the effective prize spreads they face.
Where a spread is random, its values are listed with the corresponding
probabilities in parentheses.}
\label{tab:example-battles}
\begin{tabular}{@{}cccccc@{}}
\hline
& & & \multicolumn{3}{c}{Effective prize spread $v_t$}\\
\cline{4-6}
Battle $t$ & $\mu_t$ & $\phi_t(v)$
& $\mathcal P^{\mathrm{sim}}$
& $\mathcal P^{12,3}$
& $\mathcal P^{\mathrm{seq}}$\\
\hline
$1$ & $\frac34$ & $\frac{3}{2(\rho+1)}v^\rho$
& $\frac12$ & $\frac12$ & $\frac12$\\
$2$ & $\frac12$ & $\frac{2}{\rho+1}v^\rho$
& $\frac58$ & $\frac58$
& $\frac34\,\left(\frac34\right),\ \frac14\,\left(\frac14\right)$\\
$3$ & $\frac14$ & $\frac{3}{2(\rho+1)}v^\rho$
& $\frac12$
& $1\,\left(\frac12\right),\ 0\,\left(\frac12\right)$
& $1\,\left(\frac12\right),\ 0\,\left(\frac12\right)$\\
\hline
\end{tabular}
\end{table}

Reading across the table, every battle keeps the same mean spread---%
$\frac34\cdot\frac34+\frac14\cdot\frac14=\frac58$ in battle $2$ and
$\frac12\cdot1+\frac12\cdot0=\frac12$ in battle $3$---while its
dispersion weakly increases. Since expected total output is
$TE(\mathcal P)=\sum_{t}\E\left[\phi_t(v_t)\right]$ under every design,
the three differ only in how $\phi_t$ evaluates a more or less dispersed
spread.
Table~\ref{tab:example-comparison} reports the resulting values.

\begin{table}[H]
\centering
\setlength{\tabcolsep}{7pt}
\renewcommand{\arraystretch}{1.05}
\caption{Disclosure, temporal refinement, and expected total output.}
\label{tab:example-comparison}
\begin{tabular}{@{}lccc@{}}
\hline
Technology
&
$TE(\mathcal P^{\mathrm{sim}})$
&
$TE(\mathcal P^{12,3})$
&
$TE(\mathcal P^{\mathrm{seq}})$
\\
\hline
$\gamma(y)=\sqrt{y}$ & $2.468$ & $2.261$ & $2.240$ \\
$\gamma(y)=y^2$      & $0.510$ & $0.635$ & $0.667$ \\
$\gamma(y)=y$        & $1.375$ & $1.375$ & $1.375$ \\
\hline
\end{tabular}
\end{table}

The table illustrates the common role of information refinement. Convex
output costs favor no disclosure and coarser temporal structures,
concave costs favor full disclosure and finer temporal structures, and
linear costs make both comparisons neutral.\footnote{The analysis below
generalizes these comparisons to arbitrary odd numbers of battles and
heterogeneous component battles, and ranks any two temporal structures
that are related by refinement.}

\section{Equilibrium Analysis}\label{sec:equilibrium}

This section characterizes equilibrium under the finite-support private-information specification, for an arbitrary temporal structure and disclosure policy held fixed throughout. Subsection~\ref{subsec:component-battles} treats a single component battle facing a given effective prize spread $v$, which is there a generic scalar; Subsection~\ref{subsec:overall-contest} assembles the component equilibria into an equilibrium of the overall contest, in which each $v$ is endogenously determined as a conditional expectation of $Y_t$.

\subsection{Component battles}\label{subsec:component-battles}

Consider first the linear benchmark $\gamma^0(y)=y$, in which output coincides with ability-scaled expenditure, $x=y$. A player of ability $\alpha$ who produces $y$ pays $y/\alpha$, so the component battle is a two-player all-pay auction with discrete private values and linear costs, the value of an ability type $\alpha$ being $\alpha v$. Such an auction has a unique and monotone equilibrium \citep{siegel-2014}, which admits no closed form in terms of the primitive types; the regularization procedure of \citet{chen-lu-bai-2026} supplies one, and the next lemma states it.

\begin{lemma}[Linear-benchmark equilibrium; \citealp{siegel-2014,chen-lu-bai-2026}]\label{lem:linear-equilibrium}
Fix battle $t$, let $v>0$, and let the technology be $\gamma^0(y)=y$. The component battle has a unique equilibrium, and it is monotone. Each player's primitive ability types can be subdivided into \emph{divisions} $k=1,\ldots,m_t$, where division $k$ of player $i(t)$ carries ability $\alpha_k^{i(t)}$ and probability mass $p_k^{i(t)}$, the abilities are weakly decreasing in $k$ for each player, the masses of the divisions subdividing a primitive type sum to its probability, and paired divisions carry a common mass-to-ability ratio,
\begin{equation}
\frac{p_k^{A(t)}}{\alpha_k^{A(t)}}
=
\frac{p_k^{B(t)}}{\alpha_k^{B(t)}},
\qquad
k=1,\ldots,m_t;
\label{eq:regularized-identity}
\end{equation}
any further divisions of one player are unpaired. Define
\begin{equation}
I_k^t
=
\sum_{s=k}^{m_t}p_s^{A(t)}\alpha_s^{B(t)}
=
\sum_{s=k}^{m_t}p_s^{B(t)}\alpha_s^{A(t)},
\qquad
I_{m_t+1}^t=0,
\label{eq:Ik}
\end{equation}
the two sums being equal by~\eqref{eq:regularized-identity}. In equilibrium both players randomize ability-scaled expenditure over the common range $\left[0,vI_1^t\right]$: a player in division $k$ randomizes uniformly over $\left[vI_{k+1}^t,vI_k^t\right]$, and a player in an unpaired division expends zero. Expected total output and expected aggregate expenditure are, respectively,
\begin{equation}
\frac{v}{2}\sum_{k=1}^{m_t}
\left(p_k^{A(t)}+p_k^{B(t)}\right)
\left(I_k^t+I_{k+1}^t\right)
\qquad\text{and}\qquad
v\sum_{k=1}^{m_t}
\frac{p_k^{A(t)}}{\alpha_k^{A(t)}}
\left(I_k^t+I_{k+1}^t\right).
\label{eq:linear-expenditure}
\end{equation}
Each player's expected expenditure is half the second expression, so the two players spend the same amount in expectation.
\end{lemma}

\begin{proof}
    See Appendix~\ref{app:reduction-proof}.
\end{proof}

The two expressions in~\eqref{eq:linear-expenditure} differ because a player of ability $\alpha$ spends $y/\alpha$ to produce $y$: output weights divisions by their masses, whereas expenditure weights them by mass per unit of ability. Both are proportional to $v$, which is the origin of the expenditure neutrality established in Section~\ref{sec:expenditure}.

The equal split of expenditure follows directly from the mass-to-ability identity: player $i(t)$'s expected expenditure is $\frac{v}{2}\sum_{k}\frac{p_k^{i(t)}}{\alpha_k^{i(t)}}\left(I_k^t+I_{k+1}^t\right)$, and by~\eqref{eq:regularized-identity} the summands agree across the two players term by term, while unpaired divisions spend nothing. Expected \emph{bids} are not equalized in this way: the player more likely to be drawn into a high division bids more on average, and it is only after dividing by ability that the two coincide.

Because the intervals $\left[vI_{k+1}^t,vI_k^t\right]$ are disjoint and ordered, a player in division $k$ loses to every opponent division whose interval lies above his own and wins against every division whose interval lies below it; this is what makes the equilibrium expressible in closed form. Unpaired divisions generate no output but are counted in the battle-winning probability. The representation depends only on the two ability distributions, since $v$ multiplies every private value and cancels from the normalized inverse-value coordinates, so one representation serves every effective prize spread and every admissible technology; Appendix~\ref{app:reduction-proof} gives the construction. Throughout the finite-support analysis, $\alpha_k^{i(t)}$ and $p_k^{i(t)}$ refer to these divisions unless stated otherwise.

\textbf{Illustration.}\quad
Suppose player $A(t)$ has primitive types $(\alpha,p)=(2,0.4)$ and $(1,0.6)$, and player $B(t)$ has $(3,0.3)$ and $(1,0.7)$. Assigning each primitive type the inverse-ability length $p/\alpha$ gives breakpoints $(0,0.2,0.8)$ and $(0,0.1,0.8)$, whose common refinement has three paired divisions:
\[
\begin{array}{ccc}
\hline
k & \left(\alpha_k^{A(t)},p_k^{A(t)}\right) & \left(\alpha_k^{B(t)},p_k^{B(t)}\right)\\
\hline
1 & (2,0.2) & (3,0.3)\\
2 & (2,0.2) & (1,0.1)\\
3 & (1,0.6) & (1,0.6)\\
\hline
\end{array}
\]
Each row satisfies~\eqref{eq:regularized-identity}, since $p_k^{A(t)}/\alpha_k^{A(t)}=p_k^{B(t)}/\alpha_k^{B(t)}$ equals $0.1$, $0.1$, and $0.6$ for $k=1,2,3$. Hence $m_t=3$, and~\eqref{eq:Ik} gives $I_3^t=0.6$, $I_2^t=0.8$, and $I_1^t=1.4$: the two players in division $k$ randomize uniformly over the same interval, and the three intervals $[0,0.6v]$, $[0.6v,0.8v]$, and $[0.8v,1.4v]$ tile the common bidding range. Total output is $1.12v$ and aggregate expenditure $0.72v$, the latter split evenly at $0.36v$ per player even though their expected bids are $0.54v$ and $0.58v$.

Now let the technology be a general $\gamma$, and write $y=\gamma^{-1}(x)$ for ability-scaled expenditure. Because $\gamma$ is strictly increasing, the player with the higher $y$ also produces the higher output and wins the battle, and a player of ability $\alpha$ who chooses $y$ still pays $y/\alpha$. The technology therefore does not affect the equilibrium distribution of ability-scaled expenditure; it affects only how that expenditure is converted into output.

\begin{lemma}[Equilibrium and output response]\label{lem:linear-reduction}
Fix battle $t$ and let $v\geq0$ be the effective prize spread faced by both players.
\begin{enumerate}[leftmargin=*,label=(\roman*)]
    \item \emph{(Strategic equivalence.)} The map $e\mapsto y=\alpha e$ is a bijection under which the component battle is strategically equivalent to the linear benchmark. Hence the component battle has a unique equilibrium, and its distribution of ability-scaled expenditure is the one described in Lemma~\ref{lem:linear-equilibrium}, independently of $\gamma$. In particular, each player's expected expenditure is the one given in Lemma~\ref{lem:linear-equilibrium} for every admissible technology, so the two players again spend the same amount in expectation.
    \item \emph{(Output response.)} Expected total output is
    \begin{equation}
    \phi_t(v)=
    \sum_{i\in\{A,B\}}\sum_{k=1}^{m_t}
    \frac{1}{\alpha_k^{-i(t)}}
    \int_{I_{k+1}^t}^{I_k^t}\gamma(vz)\,dz,
    \qquad v\geq0,
    \label{eq:phiI}
    \end{equation}
    and $\phi_t(0)=0$.
    \item \emph{(Level invariance.)} For $v>0$, team $A$'s equilibrium winning probability does not depend on $v$.
\end{enumerate}
\end{lemma}

\begin{proof}
    See Appendix~\ref{app:reduction-proof}.
\end{proof}

Lemma~\ref{lem:linear-reduction} is the source of the model's tractability: all strategic analysis takes place in the linear benchmark, and the curvature of $\gamma$ bears only on how equilibrium ability-scaled expenditures are converted into output. Only two ingredients are imported: uniqueness and monotonicity from \citet{siegel-2014}, and the closed-form divisional representation from \citet{chen-lu-bai-2026}. The output formula~\eqref{eq:phiI} and the level invariance in part~(iii) are established here. Part~(i) also carries a complete-information property over to private information: Corollary~1 of \citet{kaplan-wettstein-2006} shows that in a two-player complete-information all-pay auction each contestant's expected expenditure is invariant to the common cost function, and Lemma~\ref{lem:linear-reduction} establishes the same invariance under privately known abilities, with the common value given in closed form by Lemma~\ref{lem:linear-equilibrium}.

Let $\mu_t=\Prob(W_t=A)$ denote team $A$'s equilibrium winning probability in battle $t$ at a positive effective prize spread. By Lemma~\ref{lem:linear-reduction}, $\mu_t$ does not depend on the level of that spread. Moreover $\mu_t\in(0,1)$: in the divisional equilibrium of Appendix~\ref{app:reduction-proof}, the active divisions of the two players share the common support $[0,vI_1^t]$ and are atomless above zero, so each team wins with positive probability.

The curvature of $\phi_t$ follows directly from~\eqref{eq:phiI}: for each $z>0$ the map $v\mapsto\gamma(vz)$ has the same curvature as $\gamma$, and $\phi_t$ is a positive weighted integral of such maps. Equivalently, for $v>0$ wherever differentiation under the integral sign is valid,
\begin{equation}
\phi_t''(v)=
\sum_{i\in\{A,B\}}\sum_{k=1}^{m_t}
\frac{1}{\alpha_k^{-i(t)}}
\int_{I_{k+1}^t}^{I_k^t}z^2\gamma''(vz)\,dz.
\label{eq:phiI-second}
\end{equation}
Hence $\phi_t$ inherits the concavity, convexity, or linearity of $\gamma$. Since the output cost is $\gamma^{-1}(x)/\alpha$, convex output costs generate a concave output response, concave output costs a convex response, and linear output costs a linear response.

The design results below rest on two linked properties of the component
battle: $\mu_t$ is independent of the level of the effective prize spread,
and the distribution of normalized ability-scaled expenditure $y/v$ is
also level-invariant. The first keeps the outcome distribution fixed across
designs. The second implies both that the output response $\phi_t(v)$ has
the curvature of $\gamma$ and that nested information structures can be
ordered in convex order at the player level.

\begin{remark}[Scope of the component-battle properties]\label{rem:scale-invariance-scope}
Both properties stem from the multiplicative technology $x=\gamma(\alpha e)$. By Lemma~\ref{lem:linear-reduction}, each component battle is strategically an all-pay auction with linear costs and values proportional to $v$; that auction is homogeneous in $v$, so equilibrium ability-scaled expenditures scale proportionally with the spread and the winning probability is level-free. This scale invariance is what makes battle outcomes independent of disclosed histories and sustains the martingale identities~\eqref{eq:mean-preserving-disclosure} and~\eqref{eq:temporal-refinement-mps}. Under nonmultiplicative heterogeneity---for example, a player-specific output cost $c(x;a)=C(x+a)-C(a)$---the equilibrium winning probability may vary with the level of the spread. Disclosed histories would then affect not only the effective prize spreads but also the winning probabilities of later battles, and neither the disclosure comparisons nor the temporal-refinement results would apply in their present form.
\end{remark}

\subsection{The overall contest}\label{subsec:overall-contest}

We now assemble the component equilibria into an equilibrium of the overall contest. At a zero effective prize spread, both players produce zero output, and we record team $A$ as the winner with probability $\mu_t$.\footnote{This convention preserves the history-independent winning probability and hence the product-form spread formulas in Section~\ref{sec:fixed-disclosure}. It is payoff-irrelevant: a zero spread requires $\pi_t=0$ and that one team have already secured $n+1$ wins, so later recorded outcomes cannot affect the team prize. Other recording rules leave effective spreads unchanged but generally destroy the product-measure representation through $\Theta_A$.} Team $A$ therefore wins battle $t$ with probability $\mu_t$ at every public history, whatever the spread. Because types are drawn independently and players randomize independently, outcomes of distinct battles in the same block are independent; proceeding block by block, the outcome vector is a product of independent Bernoulli draws with parameters $(\mu_t)_{t\in\mathcal T}$.

For any subset $S\subseteq\mathcal T$, let $\Theta_A(r;S)$ denote the probability that team $A$ wins exactly $r$ of the battles in $S$:
\begin{equation}
\Theta_A(r;S)
=
\Prob\left(
\sum_{s\in S}\mathbf 1\{W_s=A\}=r
\right),
\label{eq:theta}
\end{equation}
where battle $s$ is won by team $A$ independently with probability $\mu_s$. We set $\Theta_A(r;S)=0$ whenever $r<0$ or $r>|S|$. Since $\mu_s\in(0,1)$ for every $s$, it follows that $\Theta_A(r;S)>0$ whenever $0\leq r\leq|S|$.

\begin{proposition}[Assembly and uniqueness of equilibrium outcomes]\label{prop:pbe-assembly}
Fix an ordered temporal structure and a committed, public, and common disclosure policy concerning previously resolved outcomes.
\begin{enumerate}[leftmargin=*]
    \item The overall contest has a perfect Bayesian equilibrium in which, at every public history, each current battle is played according to the component-battle equilibrium associated with its effective prize spread, with randomizations independent across battles.

    \item In every perfect Bayesian equilibrium, the recorded outcome vector $(W_t)_{t\in\mathcal T}$ has the same product distribution. At every on-path public history (almost surely when the signal space is continuous), each effective prize spread and the conditional distribution of output in every current battle are the same across equilibria. In particular, expected total output does not depend on the equilibrium selected.
\end{enumerate}
\end{proposition}
\begin{proof}
See Appendix~\ref{app:pbe-assembly-proof}.
\end{proof}

Part~2 combines uniqueness of the component-battle equilibrium at a positive spread with the recording convention at a zero spread. Together these fix the battle-winning probability at $\mu_t$ at every public history, whatever the equilibrium, so the expected total output of a design is a well-defined number and the comparisons in Sections~\ref{sec:fixed-disclosure} and~\ref{sec:temporal-refinement} rank equilibrium outcomes rather than selected equilibria.

\section{Information Refinement and Convex Order}
\label{sec:information-refinement}

The disclosure and temporal-structure comparisons are applications of a
single information-refinement result. By
Proposition~\ref{prop:pbe-assembly}, the recorded outcome vector has the
same product distribution across all admissible designs, so those designs
can be evaluated on one probability space. An admissible information
profile $\boldsymbol{\mathcal F}=(\mathcal F_t)_{t\in\mathcal T}$ specifies
the public information available when each battle is played: $\mathcal F_t$
is generated by committed, public, and common signals that depend only on
outcomes of battles resolved in blocks preceding $\ell(t)$ under the
associated temporal structure and, for stochastic disclosure policies, on
the designer's randomization devices. These devices are independent of
all model primitives and equilibrium randomizations, and the probability
space is enlarged accordingly. Write
\[
V_t^{\mathcal F}=\E[Y_t\mid\mathcal F_t]
\]
for the effective prize spread under the profile. Expected total output is
then
\begin{equation}
TE(\boldsymbol{\mathcal F})
=\sum_{t\in\mathcal T}
\E\left[\phi_t\left(V_t^{\mathcal F}\right)\right].
\label{eq:design-as-information}
\end{equation}

Say that $\boldsymbol{\mathcal G}$ refines
$\boldsymbol{\mathcal F}$ if
$\mathcal F_t\subseteq\mathcal G_t$ for every battle $t$. The next theorem
gives both the first-moment and distributional implications of such a
refinement.

\begin{theorem}[Information refinement, convex order, and output]
\label{thm:information-refinement}
Suppose the admissible information profile $\boldsymbol{\mathcal G}$
refines $\boldsymbol{\mathcal F}$. Then:
\begin{enumerate}[leftmargin=*]
    \item For every battle $t$ and contestant $i(t)$, equilibrium
    ability-scaled expenditures can be coupled so that
    \begin{equation}
    y_{\mathcal F}^{i(t)}
    =
    \E\!\left[
    y_{\mathcal G}^{i(t)}
    \mid
    \mathcal F_t,\alpha^{i(t)},Z^{i(t)}
    \right],
    \label{eq:expenditure-martingale-coupling}
    \end{equation}
    where $Z^{i(t)}$ is the contestant's normalized equilibrium
    ability-scaled expenditure. Thus
    $y_{\mathcal G}^{i(t)}$ is a mean-preserving spread of
    $y_{\mathcal F}^{i(t)}$. The same conclusion holds for actual
    expenditure
    $e_{\mathcal H}^{i(t)}=y_{\mathcal H}^{i(t)}/\alpha^{i(t)}$,
    $\mathcal H\in\{\mathcal F,\mathcal G\}$.

    \item The two profiles induce the same expected expenditure for every
    contestant in every battle and hence the same expected aggregate
    expenditure.

    \item If output costs are convex, then
    $TE(\boldsymbol{\mathcal F})\geq
    TE(\boldsymbol{\mathcal G})$. If output costs are concave, the
    inequality is reversed. If output costs are linear, the two profiles
    induce the same expected total output.
\end{enumerate}
The output inequality is strict if, for at least one battle $t$,
$V_t^{\mathcal G}$ is nondegenerate conditional on $\mathcal F_t$ with
positive probability and $\phi_t$ is strictly curved in the relevant
direction over the induced range of effective prize spreads.
\end{theorem}

\begin{proof}
Fix $t$ and $i(t)$. By Lemma~\ref{lem:linear-reduction}, the joint
distribution of the contestant's ability $\alpha^{i(t)}$ and normalized
ability-scaled expenditure $Z^{i(t)}:=y^{i(t)}/v$ is independent of every
positive effective prize spread $v$. Draw this pair independently of the
public history and use the same draw under both profiles. Equilibrium
ability-scaled expenditure can then be represented as
\[
y_{\mathcal F}^{i(t)}=V_t^{\mathcal F}Z^{i(t)},
\qquad
y_{\mathcal G}^{i(t)}=V_t^{\mathcal G}Z^{i(t)},
\]
with both expressions equal to zero when the corresponding spread is
zero. Because
$V_t^{\mathcal F}=\E[V_t^{\mathcal G}\mid\mathcal F_t]$ by the tower
property, and battle-$t$ types and randomizations are independent of the
history determining the spreads,
\[
\E\!\left[
y_{\mathcal G}^{i(t)}
\mid
\mathcal F_t,\alpha^{i(t)},Z^{i(t)}
\right]
=
Z^{i(t)}\E[V_t^{\mathcal G}\mid\mathcal F_t]
=
y_{\mathcal F}^{i(t)}.
\]
Dividing by the positive realized ability gives the corresponding
martingale coupling for actual expenditure. Iterated expectations yield
the expenditure equalities in part~2. Finally,
$V_t^{\mathcal F}=\E[V_t^{\mathcal G}\mid\mathcal F_t]$, and $\phi_t$
has the same curvature as $\gamma$. Conditional Jensen's inequality,
summed over battles, therefore proves part~3; the stated strictness
condition is the corresponding strict-Jensen condition.
\end{proof}

The theorem clarifies why expenditure neutrality alone does not make the
output comparisons immediate. Equality of expected aggregate expenditure
does not rank a sum of nonlinear transformations of individual
expenditures. The economically operative restriction is the
player-by-player martingale coupling in
\eqref{eq:expenditure-martingale-coupling}, which follows from nested
information and level-invariant normalized expenditure. Beyond the
maintained independence assumptions, the proof uses only two properties
of the component game: outcome probabilities are
independent of the level of the effective prize spread, and normalized
ability-scaled expenditure has a level-invariant distribution. It
therefore applies whenever the component game delivers those properties.
Subsection~\ref{subsec:continuous-types} verifies them for smooth
continuous types, and Subsection~\ref{subsec:homogeneous-contests} derives
them from degree-zero homogeneity of a general contest technology.

\section{Outcome Disclosure}\label{sec:fixed-disclosure}

This section applies Theorem~\ref{thm:information-refinement} to compare
no prior-outcome disclosure and full prior-outcome disclosure for a fixed
ordered temporal structure $\mathcal P$. Disclosure expands public
information at fixed decision dates; temporal refinement, studied next,
creates additional decision and disclosure dates.

\subsection{Effective prize spreads}\label{subsec:disclosure-prize-spreads}

Under no prior-outcome disclosure, players in battle $t$ observe no outcomes from earlier blocks. Therefore, the effective prize spread in battle $t$ is the unconditional expectation of $Y_t$:
\begin{equation}
V_t^N
=\E[Y_t]
=\pi_t+V\Theta_A(n;\mathcal T\setminus\{t\}).
\label{eq:VtN}
\end{equation}
By~\eqref{eq:theta}, $\Theta_A(n;\mathcal T\setminus\{t\})$ is the ex ante probability that battle $t$ is pivotal for the allocation of the team prize. The expected total output in battle $t$ under no prior-outcome disclosure is therefore
\begin{equation}
TE_t^N=\phi_t(V_t^N).
\label{eq:TEN}
\end{equation}

Now consider full prior-outcome disclosure. Let $\mathcal P_\ell^-:=\bigcup_{q<\ell}\mathcal B_q$ denote the set of battles in preceding blocks. Before block $\ell$, the publicly observed outcome history is the random vector $\mathcal H_\ell^{\mathcal P}=(W_s)_{s\in\mathcal P_\ell^-}$, which generates the public information $\sigma$-field $\mathcal F_\ell^{\mathcal P}=\sigma(\mathcal H_\ell^{\mathcal P})$.

For a realized history $h_\ell=(w_s)_{s\in\mathcal P_\ell^-}$, let $N_\ell^A(h_\ell)=\sum_{s\in\mathcal P_\ell^-}\mathbf 1\{w_s=A\}$ denote the number of previously disclosed wins by team $A$, and for $t\in\mathcal B_\ell$ define the set of unresolved battles other than $t$ by $R_{\ell,t}^{\mathcal P}=\mathcal T\setminus(\mathcal P_\ell^-\cup\{t\})$, which contains the other battles in the current block as well as all battles in later blocks. Under full disclosure, the effective prize spread in battle $t$ conditional on history $h_\ell$ is
\begin{equation}
V_t^{F,\mathcal P}(h_\ell)
=
\E\left[
Y_t
\mid
\mathcal H_\ell^{\mathcal P}=h_\ell
\right]
=
\pi_t
+
V\Theta_A\!\left(
n-N_\ell^A(h_\ell);
R_{\ell,t}^{\mathcal P}
\right).
\label{eq:VtF}
\end{equation}
Here, $\Theta_A\!\left(n-N_\ell^A(h_\ell); R_{\ell,t}^{\mathcal P}\right)$ is the conditional probability that battle $t$ is pivotal given history $h_\ell$: after team $A$ has accumulated $N_\ell^A(h_\ell)$ wins in preceding blocks, it must win exactly $n-N_\ell^A(h_\ell)$ of the other unresolved battles for the outcome of battle $t$ to determine the winner of the team prize. The expected total output in battle $t$ under full prior-outcome disclosure is therefore
\begin{equation}
TE_t^{F}(\mathcal P)
=
\E\left[
\phi_t\left(
V_t^{F,\mathcal P}
\left(\mathcal H_{\ell(t)}^{\mathcal P}\right)
\right)
\right].
\label{eq:TEF}
\end{equation}

\subsection{Disclosure ranking}\label{subsec:disclosure-ranking}

The key relation between the two disclosure policies follows from the law of iterated expectations:
\begin{equation}
V_t^N=\E\left[V_t^{F,\mathcal P}\left(\mathcal H_{\ell(t)}^{\mathcal P}\right)\right].
\label{eq:mean-preserving-disclosure}
\end{equation}
Thus, relative to no disclosure, full prior-outcome disclosure induces a mean-preserving spread of the effective prize spread. It changes the dispersion of the incentive faced by the players, but not its mean.

Define total expected output under no disclosure and full disclosure, respectively, by
\[
TE^N=\sum_{t\in\mathcal T}TE_t^N,
\qquad
TE^F(\mathcal P)=\sum_{t\in\mathcal T}TE_t^F(\mathcal P).
\]

\begin{proposition}[Disclosure and expected output]
\label{prop:disclosure-ranking}
Fix any ordered temporal structure
$\mathcal P=(\mathcal B_1,\ldots,\mathcal B_L)$. Then:
\begin{enumerate}[leftmargin=*]
    \item If output costs are convex, then $TE^N\geq TE^F(\mathcal P)$, so no prior-outcome disclosure is weakly optimal between the two policies.

    \item If output costs are concave, then $TE^N\leq TE^F(\mathcal P)$, so full prior-outcome disclosure is weakly optimal between the two policies.

    \item If output costs are linear, then $TE^N=TE^F(\mathcal P)$.
\end{enumerate}
The inequalities in parts~1 and~2 are strict if, for at least one battle, the full-disclosure effective prize spread is nondegenerate and the corresponding output response is strictly curved in the relevant direction over its range of effective prize spreads.
\end{proposition}

\begin{proof}
Take the coarse profile in Theorem~\ref{thm:information-refinement} to be
the uninformative no-disclosure profile and the fine profile to be
$(\mathcal F_{\ell(t)}^{\mathcal P})_{t\in\mathcal T}$. The nesting is
immediate, and the martingale relation is
\eqref{eq:mean-preserving-disclosure}, so
the theorem gives the three rankings and the stated strictness conditions.
\end{proof}

Intuitively, disclosure reallocates incentive intensity across histories without changing its average. Under no disclosure, players evaluate a single ex ante effective prize spread because they cannot distinguish high- from low-pivotality histories. Full disclosure makes incentives state contingent, raising the spread when the battle is more likely to matter and lowering it when it is less likely to matter. Whether this reallocation is desirable depends on how output responds to stronger incentives. With convex output costs the output response is concave, so the output gained by raising the spread in high-pivotality histories is outweighed by the output lost in low-pivotality histories; smoothing incentives is therefore beneficial, favoring no disclosure. With concave output costs the response is convex and the opposite force prevails: concentrating incentives in high-pivotality histories yields a net gain, favoring full disclosure. With linear output costs the gains and losses exactly offset.

Theorem~\ref{thm:information-refinement} shows why this conclusion is
stronger than expenditure neutrality: full disclosure generates a
player-by-player martingale spread, which supplies the ordering required
by Jensen's inequality.

Proposition~\ref{prop:disclosure-ranking} compares the two extremes, and the ranking extends to a broader class of committed \emph{public and common garblings of previously resolved outcomes}. Fix $\mathcal P$, and suppose that before each block $\ell$ both players in every battle of that block observe the same public signal
\[
S_\ell
=
s_\ell\!\left(
\mathcal H_\ell^{\mathcal P},
S_1,\ldots,S_{\ell-1},
\xi_\ell
\right),
\]
where the preceding-signal list is empty for $\ell=1$. The exogenous randomization device $\xi_\ell$ is independent of all model primitives and equilibrium randomizations, and the measurable signal rules $s_1,\ldots,s_L$ are publicly announced and committed to ex ante. The signal may therefore depend on resolved outcomes, earlier signals, and independent noise, but not on unresolved battle outcomes.\footnote{For example, the designer could disclose only whether team $A$ has accumulated at least $r$ wins so far, pooling histories with different values of $N_\ell^A$; or, as a stochastic garbling, reveal the history with probability $q$ and send an uninformative message with probability $1-q$, independently of the realized history.} Since players observe and recall all past signals, their public information when block $\ell$ is played is $\mathcal G_\ell=\sigma(S_1,\ldots,S_\ell)$. For $t\in\mathcal B_\ell$, write $V_t^S=\E[Y_t\mid\mathcal G_\ell]$ for the induced effective prize spread.

\begin{corollary}[Optimality over public outcome garblings]\label{cor:garbling}
Fix any ordered temporal structure $\mathcal P$ and any committed, public, and common garbling of previously resolved outcomes in the class defined above, and let $TE^S(\mathcal P)$ denote the induced expected total output under the corresponding signal policy. Then:
\begin{enumerate}[leftmargin=*]
    \item If output costs are convex, then $TE^N\geq TE^S(\mathcal P)\geq TE^F(\mathcal P)$.
    \item If output costs are concave, then $TE^N\leq TE^S(\mathcal P)\leq TE^F(\mathcal P)$.
    \item If output costs are linear, then $TE^N=TE^S(\mathcal P)=TE^F(\mathcal P)$.
\end{enumerate}
\end{corollary}

\begin{proof}
    See Appendix~\ref{app:garbling}.
\end{proof}

By construction, every policy in this class is a garbling of full
disclosure and is at least as informative as no disclosure, consistent
with the Blackwell comparison of experiments \citep{blackwell-1953}.
Because battle-winning probabilities are invariant to both the level of
any positive effective prize spread and the public history
(Subsections~\ref{subsec:component-battles} and~\ref{subsec:overall-contest}), disclosure does not alter the
distribution of battle outcomes; it only changes how finely players
distinguish among histories when assessing pivotality. Consequently, more informative disclosure increases the dispersion of
each battle's effective prize spread without changing its mean, which
places every intermediate garbling between the no-disclosure and
full-disclosure extremes. Corollary~\ref{cor:garbling} therefore follows
by applying Theorem~\ref{thm:information-refinement} twice on the enlarged
probability space: first from no disclosure to the garbling and then from
the garbling to the $\sigma$-field generated by the full history and the
designer's independent randomization. Conditioning on the latter produces
the same effective prize spread as full disclosure.

\section{Temporal Refinement}
\label{sec:temporal-refinement}

The preceding section evaluates prior-outcome disclosure policies under a fixed temporal structure. We now fix the disclosure policy and study how refining the temporal structure---splitting blocks under full disclosure---affects expected total output.

\subsection{Reduction to full disclosure}\label{subsec:temporal-as-info}

Under no prior-outcome disclosure, every ordered temporal structure generates the same effective prize spreads, because previously resolved outcomes are never observed. This common benchmark is replicated under full prior-outcome disclosure by the simultaneous structure $\mathcal P^{\mathrm{sim}}=(\mathcal T)$, which has a single block and therefore no preceding outcome to disclose. Every no-disclosure temporal structure is thus output-equivalent to $\mathcal P^{\mathrm{sim}}$ under full disclosure, and it suffices to conduct the temporal comparison within the full-disclosure family.

Given a full-disclosure temporal structure $\mathcal P=(\mathcal B_1,\ldots,\mathcal B_L)$, let $\mathcal F_t^{\mathcal P}:=\mathcal F_{\ell(t)}^{\mathcal P}$ denote the public information available when battle $t$ is played, which under full disclosure consists of the outcomes of all battles resolved in preceding blocks. The effective prize spread in battle $t$ is $V_t^{\mathcal P}=\E[Y_t\mid\mathcal F_t^{\mathcal P}]$,\footnote{Because full prior-outcome disclosure is maintained throughout this section, $V_t^{\mathcal P}$ is shorthand for $V_t^{F,\mathcal P}$.} and the expected total output under $\mathcal P$ is
\begin{equation}
TE(\mathcal P)
=
\sum_{t\in\mathcal T}
\E\left[\phi_t\left(V_t^{\mathcal P}\right)\right].
\label{eq:temporal-objective}
\end{equation}

\subsection{Refinement ranking}\label{subsec:refinement-ranking}

\begin{definition}[Temporal refinement]\label{def:temporal-refinement}
Let $\mathcal P=(\mathcal B_1,\ldots,\mathcal B_L)$ and $\mathcal P'=(\mathcal C_1,\ldots,\mathcal C_M)$ be two ordered partitions of $\mathcal T$. We say that $\mathcal P'$ is a \emph{temporal refinement} of $\mathcal P$ if there exist integers $0=j_0<j_1<\cdots<j_L=M$ such that $\mathcal B_\ell=\bigcup_{q=j_{\ell-1}+1}^{j_\ell}\mathcal C_q$ for every $\ell=1,\ldots,L$. If $M>L$, the refinement is strict.
\end{definition}

A temporal refinement preserves the order of the original blocks but may split some of them into smaller sub-blocks.\footnote{For example, $(\{1\},\{2\},\{3\})$ is a strict refinement of $(\{1,2\},\{3\})$, with $j_0=0$, $j_1=2$, and $j_2=3$.} Splitting a block creates additional disclosure dates, so the public information available before any subsequent battle weakly expands: $\mathcal F_t^{\mathcal P}\subseteq\mathcal F_t^{\mathcal P'}$ for every battle $t$. By the law of iterated expectations,
\begin{equation}
V_t^{\mathcal P}
=
\E\left[V_t^{\mathcal P'}\mid \mathcal F_t^{\mathcal P}\right].
\label{eq:temporal-refinement-mps}
\end{equation}
A refinement therefore lets players distinguish previously pooled histories, inducing a mean-preserving spread of each battle's effective prize spread. This extends the informational logic of disclosure to temporal structures.

\begin{proposition}[Temporal refinement and expected output]
\label{prop:temporal-refinement}
Suppose $\mathcal P'$ is a refinement of $\mathcal P$. Then:
\begin{enumerate}[leftmargin=*]
    \item If output costs are convex, then $TE(\mathcal P)\geq TE(\mathcal P')$, so temporal refinement weakly reduces expected total output.

    \item If output costs are concave, then $TE(\mathcal P)\leq TE(\mathcal P')$, so temporal refinement weakly increases expected total output.

    \item If output costs are linear, then $TE(\mathcal P)=TE(\mathcal P')$.
\end{enumerate}
The inequalities in parts~1 and~2 are strict if, for at least one battle $t$, $V_t^{\mathcal P'}$ is nondegenerate conditional on $\mathcal F_t^{\mathcal P}$ with positive probability and the corresponding output response is strictly curved in the relevant direction over its range of effective prize spreads.
\end{proposition}

\begin{proof}
For every battle $t$,
$\mathcal F_t^{\mathcal P}\subseteq
\mathcal F_t^{\mathcal P'}$, so the information profile under
$\mathcal P'$ refines the profile under $\mathcal P$.
Equation~\eqref{eq:temporal-refinement-mps} is the associated martingale
relation. Theorem~\ref{thm:information-refinement} gives the three
rankings, with strictness under the stated strict-Jensen conditions.
\end{proof}

Proposition~\ref{prop:temporal-refinement} extends the preceding disclosure
logic to temporal structure. By making additional resolved outcomes
available before later battles, a temporal refinement induces not only a
mean-preserving spread of the corresponding effective prize spreads but,
by Theorem~\ref{thm:information-refinement}, a player-by-player
mean-preserving spread of equilibrium ability-scaled expenditure. The
curvature of the conversion technology then determines whether splitting
blocks raises or lowers expected total output. Thus, disclosure and
temporal refinement operate through the same pivotality-information
channel, and the comparison does not require symmetry across component
battles.

\begin{corollary}[Unrestricted choice of temporal structure]
\label{cor:unrestricted-clustering}
Suppose the designer may choose any ordered partition of the fixed set
of battles and maintains full disclosure between blocks.
\begin{enumerate}[leftmargin=*]
    \item If output costs are convex, the simultaneous structure
    $\mathcal P^{\mathrm{sim}}=(\mathcal T)$ maximizes expected total output over
    all ordered partitions.
    \item If output costs are concave, at least one maximizer is a
    singleton-block structure.

    \item If output costs are linear, all ordered partitions
    generate the same expected total output.
\end{enumerate}
\end{corollary}

\begin{proof}
Every ordered partition is a refinement of
$\mathcal P^{\mathrm{sim}}=(\mathcal T)$. Moreover, every nonsingleton
ordered partition has a singleton-block refinement obtained by ordering
the battles within each block while preserving the order of the
original blocks. The conclusions follow from
Proposition~\ref{prop:temporal-refinement}.
\end{proof}

\section{Expenditure Neutrality}\label{sec:expenditure}

Theorem~\ref{thm:information-refinement} already implies equal expected
expenditure for designs ordered by information refinement. This section
establishes a broader neutrality result: expected aggregate expenditure is
the same even across temporal structures whose information sets are not
nested. It also gives the common expenditure level in closed form.

One form of invariance is already in hand:
Proposition~\ref{prop:pbe-assembly} fixes the effective prize spread at
every on-path public history and pins down play in every battle, so
expected aggregate expenditure is the same in every equilibrium of a
given design. The next result shows that expenditure is also invariant
\emph{across} designs---under no disclosure, under full disclosure, under
every committed public garbling, and under every ordered temporal
structure---and gives the common value in closed form.

\begin{proposition}[Expenditure neutrality]
\label{prop:expenditure-neutrality}
Fix any admissible technology $\gamma$. Expected aggregate expenditure is the
same under no prior-outcome disclosure, under full prior-outcome
disclosure, and under every committed public garbling in the class of
Corollary~\ref{cor:garbling}, and it is the same for every ordered
temporal structure. Its value is
\[
\sum_{t\in\mathcal T}
c_t\bigl(\pi_t+V\Theta_A(n;\mathcal T\setminus\{t\})\bigr),
\]
where $c_t\geq0$ depends only on the primitives of battle $t$.
\end{proposition}

\begin{proof}
    See Appendix~\ref{app:expenditure}.
\end{proof}

The mechanism is a linearity that output does not share. Equilibrium
ability-scaled expenditure in a component battle scales proportionally
with the effective prize spread, so expected expenditure in battle $t$ is
linear in that spread. More informative disclosure or a finer temporal
structure raises the spread after some histories and lowers it after
others without changing its unconditional mean, and a linear function of
the spread is insensitive to that redistribution. Output responds to the
spread through $\phi_t$, which inherits the curvature of $\gamma$ and is
linear in the linear benchmark. This curvature difference is the source
of the rankings in
Sections~\ref{sec:fixed-disclosure} and~\ref{sec:temporal-refinement}.

Proposition~\ref{prop:expenditure-neutrality} is only
a first-moment statement. Equality of expected aggregate expenditure does
not itself rank
$\sum_{t\in\mathcal T}\sum_{i\in\{A,B\}}
\gamma(\alpha^{i(t)}e^{i(t)})$, because total output is a sum of
player-level nonlinear transformations rather than a transformation of
aggregate expenditure. The disclosure and refinement rankings use the
stronger martingale coupling in
Theorem~\ref{thm:information-refinement}; for nonnested temporal structures,
that coupling generally does not exist and expenditure neutrality alone
has no directional implication for output.

Expenditure neutrality is therefore not what separates our results from
\citet{fu-lu-pan-2015}. The distinction concerns the response of the
outcome-determining variable to the effective prize spread.
In their model, effort---the counterpart of output in our
model---determines the outcome of a component battle and has linear cost.
Equilibrium effort is therefore proportional to the common prize spread,
so expected total effort is invariant to changes in its dispersion. In
our model, the outcome-determining variable is output $x$, whose cost
$\gamma^{-1}(x)/\alpha$ is generally nonlinear. The output response
$\phi_t(v)$ is then generally nonlinear in the prize spread, so disclosure and
temporal refinement can change expected total output by changing the
dispersion of the prize spread while preserving its mean. At the same
time, the transformation
$y=\gamma^{-1}(x)=\alpha e$ restores linear costs, which explains why
expected aggregate expenditure remains neutral.

Disclosure and temporal refinement therefore redistribute effective prize
spreads across histories while leaving expected aggregate expenditure
unchanged. Nested redistributions admit the convex-order comparison that
ranks nonlinear output. We next examine the robustness of the key
properties underlying these results to alternative assumptions.

\section{Extensions}\label{sec:extensions}

This section extends the main analysis in three directions. First, we establish that the disclosure and temporal-refinement results hold under smooth continuous-type ability distributions. Second, we verify that the fixed-structure disclosure ranking is robust to alternative assumptions about type observability. Third, we show that the mechanism extends beyond the deterministic highest-output rule to general degree-zero contest technologies.

\subsection{Continuous-type abilities}\label{subsec:continuous-types}

The finite-support assumption is used to derive the baseline component-battle representation in Section~\ref{sec:equilibrium}; the disclosure and temporal-refinement arguments rely instead on the component-battle properties established in Subsection~\ref{subsec:component-battles}. We now establish the same properties under a smooth continuous-type specification.

Suppose player $i(t)$ draws a private quantile $\theta^{i(t)}\sim\operatorname{Unif}[0,1]$, independently across players and battles, with ability $\alpha^{i(t)}(\theta^{i(t)})$, where $\alpha^{i(t)}:[0,1]\to\mathbb R_{++}$ is continuously differentiable with $(\alpha^{i(t)})'>0$ on $[0,1]$.

Fix battle $t$ and suppose $\int_0^1(1/\alpha^{A(t)}(s))\,ds\leq\int_0^1(1/\alpha^{B(t)}(s))\,ds$; the reverse case follows by interchanging the player labels. Let $\underline{\theta}_t\in[0,1]$ be the unique value satisfying
\[
\int_0^1\frac{ds}{\alpha^{A(t)}(s)}
=
\int_{\underline{\theta}_t}^1
\frac{ds}{\alpha^{B(t)}(s)},
\]
and define the strictly increasing matching map
$G_t:[0,1]\to[\underline{\theta}_t,1]$ by
\begin{equation}
\int_{\theta}^{1}\frac{ds}{\alpha^{A(t)}(s)}
=
\int_{G_t(\theta)}^{1}
\frac{ds}{\alpha^{B(t)}(s)}.
\label{eq:continuous-matching}
\end{equation}
This map is continuously differentiable and satisfies $G_t(0)=\underline{\theta}_t$, $G_t(1)=1$, and $\alpha^{A(t)}(\theta)G_t'(\theta)=\alpha^{B(t)}(G_t(\theta))$. Aggregating matched abilities, define
\begin{equation}
K_t(\theta)
:=
\int_0^{\theta}
\alpha^{A(t)}(s)G_t'(s)\,ds
=
\int_0^{\theta}
\alpha^{B(t)}(G_t(s))\,ds.
\label{eq:continuous-K}
\end{equation}

The following proposition uses this matching construction to characterize
the component-battle equilibrium and its output response. The existence,
uniqueness, and monotonicity results of \citet{amann-leininger-1996} apply
to the induced two-player asymmetric all-pay auction;
Appendix~\ref{app:continuous-types} identifies that equilibrium directly
for the present specification, in which the two induced value supports
generally differ.

\begin{proposition}[Continuous-type component battles]
\label{prop:continuous-types}
For every effective prize spread $v>0$, the component battle has a
unique Bayesian Nash equilibrium, up to changes on null sets, and that
equilibrium is monotone.
In terms of ability-scaled expenditure, the equilibrium is represented almost everywhere by
\[
y^{A(t)}(\theta;v)=vK_t(\theta),
\qquad
y^{B(t)}(G_t(\theta);v)=vK_t(\theta),
\]
while types $\eta\in[0,\underline{\theta}_t)$ of player $B(t)$ choose
$y=0$. At $v=0$, every type uniquely chooses zero expenditure.
Expected total output is
\begin{equation}
\phi_t^{\mathrm{ct}}(v)
=
\int_0^1
\gamma\!\left(vK_t(\theta)\right)
\left[1+G_t'(\theta)\right]\,d\theta,
\qquad v\geq0.
\label{eq:continuous-phi}
\end{equation}
Under this orientation, player $A(t)$'s winning probability is $\mu_t^{\mathrm{ct}}=\int_0^1G_t(\theta)\,d\theta$, which lies in $(0,1)$ and is independent of $v>0$. Moreover, $\phi_t^{\mathrm{ct}}$ has the same curvature as $\gamma$.
\end{proposition}

\begin{proof}
See Appendix~\ref{app:continuous-types}.
\end{proof}

Proposition~\ref{prop:continuous-types} establishes the key
component-battle properties needed for the preceding analysis: the
output response inherits the curvature of $\gamma$, the battle-winning
probability is invariant to the positive effective prize spread, and
equilibrium ability-scaled expenditure scales proportionally with that spread.
Together with independent type draws and the zero-spread
outcome-recording convention, these properties allow the preceding
results to extend to continuous types. The following corollary states
this extension.

\begin{corollary}[Continuous-type extension of the design results]
\label{cor:continuous-type-extension}
Under the smooth continuous-type private-information specification of
this subsection, set
\[
(\phi_t,\mu_t)
=
(\phi_t^{\mathrm{ct}},\mu_t^{\mathrm{ct}})
\qquad
\text{for every }t\in\mathcal T.
\]
Maintaining the zero-spread outcome-recording convention of
Subsection~\ref{subsec:overall-contest}, both parts of
Proposition~\ref{prop:pbe-assembly} continue to hold.
Theorem~\ref{thm:information-refinement} and the conclusions of
Proposition~\ref{prop:disclosure-ranking},
Corollary~\ref{cor:garbling},
Proposition~\ref{prop:temporal-refinement},
and Corollary~\ref{cor:unrestricted-clustering} therefore hold under the continuous-type
specification above, with all other assumptions, qualifications, and
strictness conditions unchanged.
Proposition~\ref{prop:expenditure-neutrality} also continues to hold.
\end{corollary}

\begin{proof}
See Appendix~\ref{app:cor:continuous-type-extension}.
\end{proof}

\subsection{Type observability}\label{subsec:type-discussion}

The main private-information analysis assumes that each player observes
only his own type. We now consider two alternative type-observability
regimes and examine whether the fixed-structure disclosure ranking
continues to hold. Under \emph{local type observability}, the two players
in a component battle observe each other's types before choosing
expenditure levels but do not observe type realizations in other
battles. Under \emph{global type revelation}, realized type pairs in all
battles are publicly observed before expenditure levels are chosen. The
complete-information calculations condition on realized abilities and
therefore apply to either ability specification considered above.

For these extensions, assume $\pi_t>0$ for every battle. This restriction
ensures that every effective prize spread is strictly positive and avoids
the need for a type-dependent outcome-recording convention at a zero
spread. The two observability regimes affect the analysis in different
ways. Local type observability changes the component-battle equilibrium,
whereas global type revelation additionally changes beliefs about the
pivotality of a battle through the realized types in other battles. We
consider these regimes only for the fixed-structure disclosure
comparison; we do not compare expected total output across
type-observability regimes or extend the temporal-refinement result to
them.

Both regimes are exogenous informational environments rather than design
instruments. The designer in our model discloses outcomes, not
abilities. Disclosing abilities is a different instrument: it changes
the component-battle equilibrium itself, and hence both the
battle-winning probability $\mu_t$ and the output response $\phi_t$,
whereas outcome disclosure leaves beliefs about abilities and the
component-game primitives unchanged and moves only the effective prize
spread \citep[on ability disclosure in contests
and all-pay auctions, see][]{fu-jiao-lu-2014,zhang-zhou-2016,lu-ma-wang-2018,serena-2022}.
Corollary~\ref{cor:type-observability} shows that
the outcome-disclosure ranking is robust to which of the two alternative
type-information environments prevails; ranking expected total output
across those environments is a separate question that we do not
address.

\textbf{Local type observability.}\quad
Suppose the two players in each component battle observe each other's types before choosing expenditure levels, while type realizations in other battles remain unobserved. Conditional on the realized type pair, the component battle is a complete-information all-pay auction, and averaging over that pair gives an ex ante output response $\phi_t^C$ and battle-winning probability $\mu_t^C$. As shown in Lemma~\ref{lem:complete-allpay} and Appendix~\ref{app:observable-types}, $\phi_t^C$ has the same curvature relation to $\gamma$ as the private-information response, while $\mu_t^C$ is independent of the positive level of the effective prize spread. Replacing $(\phi_t,\mu_t)$ by $(\phi_t^C,\mu_t^C)$, full prior-outcome disclosure again induces a mean-preserving spread of the effective prize spread relative to no disclosure, and the curvature argument underlying Proposition~\ref{prop:disclosure-ranking} applies unchanged within this regime.
\medskip

\textbf{Global type revelation.}\quad
Now suppose all realized type pairs are publicly observed before expenditure levels are chosen, so that realized types in other battles affect the probability that the current battle is pivotal. Conditional on the full type profile $\boldsymbol{\alpha}$, the disclosure comparison has the same structure: under no prior-outcome disclosure the effective prize spread in battle $t$ is $\E[Y_t\mid\boldsymbol{\alpha}]$, whereas under full prior-outcome disclosure it additionally conditions on the realized outcomes of preceding blocks, and by the law of iterated expectations the former is the conditional mean of the latter. Full disclosure therefore induces a mean-preserving spread conditional on the realized type profile. Because the complete-information output response conditional on the current battle's realized type pair has the same curvature relation to $\gamma$, the same conditional mean-preserving-spread argument applies.

\begin{corollary}[Disclosure under alternative type observability]
\label{cor:type-observability}
Fix any ordered temporal structure $\mathcal P$ and suppose $\pi_t>0$ for
every $t\in\mathcal T$. Under either the finite-support or smooth
continuous-type ability specification, and under either local type
observability or global type revelation, the fixed-structure disclosure
ranking in Proposition~\ref{prop:disclosure-ranking} continues to hold. In
particular, no prior-outcome disclosure generates weakly higher expected
total output under convex output costs, full
prior-outcome disclosure generates weakly higher expected total
output under concave output costs, and the two policies are
output-equivalent under linear output costs. Within either
observability regime, expected aggregate expenditure is also the same
under no and full prior-outcome disclosure.
\end{corollary}

\begin{proof}
See Appendix~\ref{app:observable-types}.
\end{proof}

\subsection{General degree-zero contest technologies}
\label{subsec:homogeneous-contests}

The design results do not require the deterministic highest-output rule.
They require scale invariance of the component contest in the linearly
costly action $y=\alpha e=\gamma^{-1}(x)$ or, conditional on the ability
pair, equivalently in expenditure $e$. Replace the deterministic rule in
battle $t$ by contest success functions
$q_t^i(y^A,y^B;\boldsymbol{\alpha}_t)$, $i\in\{A,B\}$, with
$q_t^A+q_t^B=1$. These functions may depend on the realized ability pair.
Resolution lotteries are independent across battles and independent of
type draws and contestants' mixed-strategy randomizations. For every ability pair and
every $\lambda>0$, assume
\begin{equation}
q_t^i(\lambda y^A,\lambda y^B;\boldsymbol{\alpha}_t)
=
q_t^i(y^A,y^B;\boldsymbol{\alpha}_t),
\qquad i\in\{A,B\}.
\label{eq:degree-zero-contest}
\end{equation}
The designer continues to value output $x^i=\gamma(y^i)$; the functions
$q_t^i$ determine only the recorded winner. The rule at
$(y^A,y^B)=(0,0)$ is fixed across positive effective prize spreads; at a
zero spread, we retain the outcome-recording convention of
Subsection~\ref{subsec:overall-contest}. Conditional on his own ability,
contestant $i(t)$ then maximizes
\[
v\,\E\!\left[
q_t^i(y^{i(t)},y^{-i(t)};\boldsymbol{\alpha}_t)
\mid \alpha^{i(t)}
\right]
-
\frac{y^{i(t)}}{\alpha^{i(t)}}.
\]
Assume that at $v=1$ an equilibrium outcome distribution---the joint
distribution of realized abilities, ability-scaled expenditures, and the
recorded winner---exists and is unique across all equilibria, whether pure
or mixed. Assume also that the expectations in
Proposition~\ref{prop:homogeneous-contests} are finite.

\begin{proposition}[Degree-zero component contests]
\label{prop:homogeneous-contests}
Under~\eqref{eq:degree-zero-contest}, the equilibrium outcome distribution at
every $v>0$ is unique and is obtained by scaling each ability-scaled expenditure
in any $v=1$ equilibrium by $v$. Thus, for normalized equilibrium actions
$(Z^{A(t)},Z^{B(t)})$ whose joint distribution is independent of $v$,
\[
y_t^{i(t)}(v)=vZ^{i(t)},
\qquad i\in\{A,B\}.
\]
The equilibrium battle-winning probability is independent of $v$, the
expected-output response is
\begin{equation}
\phi_t^q(v)
=
\E\!\left[
\gamma\!\left(vZ^{A(t)}\right)
+
\gamma\!\left(vZ^{B(t)}\right)
\right],
\label{eq:homogeneous-output-response}
\end{equation}
and expected aggregate expenditure in battle $t$ is
\begin{equation}
v\kappa_t^q,
\qquad
\kappa_t^q
=
\E\!\left[
\frac{Z^{A(t)}}{\alpha^{A(t)}}
+
\frac{Z^{B(t)}}{\alpha^{B(t)}}
\right].
\label{eq:homogeneous-expenditure}
\end{equation}
Consequently, $\phi_t^q$ inherits the concavity, convexity, or linearity
of $\gamma$. Under the maintained independence assumptions and zero-spread
recording convention, the preceding equilibrium-assembly, disclosure,
temporal-refinement, and expenditure results hold with $(\phi_t,\mu_t)$ replaced
by the degree-zero-contest objects and with $c_t=\kappa_t^q$.
\end{proposition}

\begin{proof}
For $v>0$, write $y^{i(t)}=vz^{i(t)}$. Dividing contestant $i(t)$'s
payoff by $v$ and using~\eqref{eq:degree-zero-contest} gives
\[
\E\!\left[
q_t^i(z^{i(t)},z^{-i(t)};\boldsymbol{\alpha}_t)
\mid \alpha^{i(t)}
\right]
-
\frac{z^{i(t)}}{\alpha^{i(t)}},
\]
which is independent of $v$. Scaling therefore defines bijections between the
pure- and mixed-strategy equilibrium sets at $v=1$ and $v$. Uniqueness of the
joint equilibrium distribution gives the stated representation. Winning-probability
invariance and equations~\eqref{eq:homogeneous-output-response}--\eqref{eq:homogeneous-expenditure}
follow immediately.
Equation~\eqref{eq:homogeneous-output-response} is a sum of expectations
of functions $v\mapsto\gamma(vZ)$, so it has the same weak curvature as
$\gamma$. The remaining conclusions follow from the level invariance of winning
probabilities and normalized actions and from battle independence. At $v=0$,
zero expenditure is uniquely optimal; the recording convention of
Subsection~\ref{subsec:overall-contest} preserves these properties.
\end{proof}

The class includes, for example, weighted Tullock technologies
\citep{tullock-1980}
\[
q_t^A(y^A,y^B)
=
\frac{\beta_A(y^A)^r}
{\beta_A(y^A)^r+\beta_B(y^B)^r},
\qquad
q_t^B=1-q_t^A,
\]
with $\beta_A,\beta_B>0$ and $r>0$, subject to a fixed rule at the origin
and to the maintained uniqueness assumption on the equilibrium outcome
distribution. More
generally, ratio-form technologies with impact functions homogeneous of
a common degree satisfy~\eqref{eq:degree-zero-contest}; see also
\citet{skaperdas-1996}.

Uniqueness in distribution is sufficient but not necessary. The same
conclusions follow if a scale-consistent equilibrium selection is fixed
across prize spreads. Without uniqueness or such a selection, scaling
still maps the equilibrium set at $v=1$ onto the equilibrium set at $v$,
but different selections can generate different outcome and output
distributions.

The location of the homogeneity restriction matters. If a primitive
contest success function is homogeneous of degree zero only in final
outputs $x^i=\gamma(y^i)$, then its composite success function is
$q_t^i(\gamma(y^A),\gamma(y^B);\boldsymbol{\alpha}_t)$. Degree-zero
homogeneity in $x$ alone does not imply~\eqref{eq:degree-zero-contest}
for an arbitrary nonlinear $\gamma$, because
$\gamma(vz^A)$ and $\gamma(vz^B)$ need not be the same multiple of
$\gamma(z^A)$ and $\gamma(z^B)$. The implication does hold if $\gamma$
is itself positively homogeneous---as with a power technology---or if
the contest rule depends only on the ranking of outputs. The original
highest-output rule is of the latter kind: strict monotonicity of
$\gamma$ makes ranking outputs equivalent to ranking $y$. Without this
composite scale invariance, equilibrium winning probabilities and
normalized actions may vary with $v$, so the disclosure and temporal
comparisons need not survive.

\section{Conclusion}\label{sec:conclusion}

In multi-battle team contests, outcome disclosure and temporal refinement
operate through a common channel: both make each battle's effective prize
spread weakly more dispersed across histories without changing its mean.
Scale invariance transmits this information ordering to equilibrium
behavior, making ability-scaled expenditure under the finer information
structure a mean-preserving spread of its coarser counterpart player by
player. The curvature of the output technology then delivers the design
ranking by Jensen's inequality. Convex output costs favor no disclosure and
coarser temporal structures; concave costs favor full disclosure and finer
ones; and linear costs make both comparisons neutral. Within the class of
committed, public, and common garblings of previously resolved outcomes, no
and full disclosure are the relevant bounds.

Because these instruments only redistribute effective prize spreads across
histories, expected aggregate expenditure is invariant across disclosure
policies and temporal structures. Mean neutrality alone does not rank output;
the sharper player-level convex-order relation does so for nested designs. The
same logic extends to degree-zero component contest technologies in the
linearly costly action, provided their equilibrium outcome distribution is
unique. The results also hold for finite-support and smooth continuous-type
ability distributions. Under strictly positive battle-specific prizes, the
fixed-structure disclosure ranking further extends to local type observability
and global type revelation. The optimal ordering of heterogeneous battles,
which is not governed by the refinement order, remains open.

\clearpage
\appendix
\section*{Appendix}

\section{Equilibrium Proofs}\label{app:component-proof}

\subsection{Proof of Lemmas~\ref{lem:linear-equilibrium} and~\ref{lem:linear-reduction}}\label{app:reduction-proof}

\emph{Strategic equivalence.} Fix battle $t$, a realized ability $\alpha$, and a positive effective prize spread $v$. Define ability-scaled expenditure $y=\alpha e=\gamma^{-1}(x)$. Because $\gamma$ is strictly increasing, comparing outputs is equivalent to comparing $y$, so the player's payoff can be written as
\[
v\Pr(\text{win at }y)+\frac v2\Pr(\text{tie at }y)-\frac{y}{\alpha}.
\]
Multiplying by the positive constant $\alpha$ preserves best responses and gives
\[
\alpha v\Pr(\text{win at }y)+\frac{\alpha v}{2}\Pr(\text{tie at }y)-y,
\]
which is the linear benchmark: an all-pay auction with value $\alpha v$ and linear cost $y$. Since $e\mapsto y=\alpha e$ is a bijection, equilibrium strategies correspond under this transformation, and the transformed game is independent of $\gamma$. Uniqueness therefore holds for the component battle because it holds for the benchmark, the latter being a two-player all-pay auction with discrete private values and linear costs, a special case of the environment for which \citet{siegel-2014} establishes a unique equilibrium. Output is recovered from $x=\gamma(y)$, which is the only role played by the technology. This proves part~(i) of Lemma~\ref{lem:linear-reduction}.

\emph{The divisional representation.} For player $i$, let $\{\hat\alpha_j^i,\hat p_j^i\}_{j=1}^{J_t^i}$ denote the primitive ability types and set $\lambda_j^i=\hat p_j^i/\hat\alpha_j^i$ and $\Lambda^i=\sum_{j}\lambda_j^i$. Represent player $i$'s distribution by $[0,\Lambda^i]$, partitioned into consecutive intervals of lengths $\lambda_j^i$ ordered from the highest ability to the lowest; align the two players' partitions at zero and take their common refinement on $[0,\min\{\Lambda^A,\Lambda^B\}]$. If segment $k$ has length $\lambda_k$ and carries abilities $\alpha_k^{A(t)}$ and $\alpha_k^{B(t)}$, assign masses $p_k^{i(t)}=\lambda_k\alpha_k^{i(t)}$. Then
\[
\frac{p_k^{A(t)}}{\alpha_k^{A(t)}}
=
\lambda_k
=
\frac{p_k^{B(t)}}{\alpha_k^{B(t)}},
\]
which is~\eqref{eq:regularized-identity}. Any residual interval belongs to at most one player and carries abilities weakly below that player's lowest paired ability. The construction preserves the mass of every primitive type, so division-level strategies aggregate back to the original type distribution, and it depends only on the two ability distributions, hence serves every $v>0$ and every admissible $\gamma$. This is the procedure of \citet{chen-lu-bai-2026}, specialized to the values $\alpha v$.

\emph{Equilibrium and output.} Fix $v>0$; by strategic equivalence it suffices to work with $y$. For active paired division $k$ of player $i(t)$, prescribe a uniform distribution on $[vI_{k+1}^t,vI_k^t]$ with density $1/(vp_k^{i(t)}\alpha_k^{-i(t)})$. By~\eqref{eq:regularized-identity}, $p_k^{i(t)}\alpha_k^{-i(t)}=p_k^{-i(t)}\alpha_k^{i(t)}$, so the two sums in~\eqref{eq:Ik} agree and $I_k^t$ is common to the two players. Moreover $v(I_k^t-I_{k+1}^t)=vp_k^{i(t)}\alpha_k^{-i(t)}$, so the prescribed density integrates to one and both players use the same ordered support intervals. Any unpaired lower divisions choose zero.

To verify optimality, let $\eta_t^{-i}$ denote any opponent mass at zero. If division $k$ of player $i(t)$ chooses $y\in[vI_{r+1}^t,vI_r^t]$, the probability mass of the opponent's active divisions weakly below $y$ is
\[
F_{\mathrm{act}}^{-i(t)}(y)
=
\sum_{s=r+1}^{m_t}p_s^{-i(t)}
+
\frac{y/v-I_{r+1}^t}{\alpha_r^{i(t)}}.
\]
Its payoff on this interval is therefore
\[
U_k(y)
=v\eta_t^{-i}+vF_{\mathrm{act}}^{-i(t)}(y)-\frac{y}{\alpha_k^{i(t)}},
\]
with derivative
\[
U_k'(y)=\frac{1}{\alpha_r^{i(t)}}-\frac{1}{\alpha_k^{i(t)}}.
\]
Because abilities are weakly decreasing in $k$, this derivative is nonnegative below division $k$'s assigned interval, zero on that interval, and nonpositive above it. Choosing above the common support only raises cost. If the opponent has mass at zero, choosing $y=0$ collects only half of that mass and is therefore no better than the limit from positive choices. Thus every active division best responds. For an unpaired division with ability $\alpha_0\leq\alpha_{m_t}^{i(t)}$, the same derivative is nonpositive on every positive interval, so zero is optimal.

Conditional on active division $k$, expected output is
\[
\int_{vI_{k+1}^t}^{vI_k^t}
\frac{\gamma(w)}{vp_k^{i(t)}\alpha_k^{-i(t)}}\,dw.
\]
Multiplying by $p_k^{i(t)}$, summing over active divisions and players, and changing variables $w=vz$ gives~\eqref{eq:phiI}. Unpaired divisions contribute zero. Aggregation over divisions carrying the same primitive ability yields an equilibrium of the original component battle.

The distribution of $y/v$ is independent of $v$, so the battle-winning probability is independent of every positive effective prize spread. Existence and the output formula follow from the construction above; uniqueness is part~(i) of Lemma~\ref{lem:linear-reduction}. At $v=0$, zero expenditure is uniquely optimal for each player and $\phi_t(0)=0$.

\subsection{Proof of Proposition~\ref{prop:pbe-assembly}}\label{app:pbe-assembly-proof}

\emph{Part 1.} Fix a public history before a block. Unresolved types retain their independent primitive distributions. Beliefs about resolved outcomes and signals follow Bayes' rule on path; at zero-probability histories, choose beliefs consistent with the committed recording and disclosure rules. For each current battle $t$, let $v_t=\E[Y_t\mid\text{public history}]$. Prescribe the component-battle equilibrium of Lemma~\ref{lem:linear-reduction} when $v_t>0$, and zero expenditure together with the maintained recording convention when $v_t=0$, using independent randomizations across battles. Each contestant participates in only one battle, so conditional on the public history his gain from winning rather than losing is exactly $v_t$; hence his deviation problem is precisely the corresponding component battle. The prescribed action is therefore optimal. Scale invariance makes the induced winning probability $\mu_t$ at every positive spread, and the recording convention gives the same probability at zero, so the continuation probabilities used to compute $v_t$ are consistent with the prescribed strategies. This establishes sequential rationality and consistent beliefs.

\emph{Part 2.} Fix any perfect Bayesian equilibrium and any on-path public history at which battle $t$ is played. Whatever the effective prize spread $v_t$ at that history, team $A$ is recorded as the winner of battle $t$ with probability $\mu_t$. If $v_t>0$, then, since each contestant takes part in a single battle, the game faced by the two contestants in battle $t$ is the component battle with common spread $v_t$; its equilibrium is unique by Lemma~\ref{lem:linear-reduction}, and the winning probability at any positive spread is $\mu_t$. If $v_t=0$, both contestants choose zero expenditure and the recording convention assigns the same probability. The conclusion therefore does not require knowing which of the two cases obtains.

Distinct contestants randomize independently, and types are independent across players and battles, so, conditional on any public history, the outcomes of battles in the same block are independent. Proceeding block by block, the recorded outcome vector is a product of independent Bernoulli distributions with parameters $\mu_t$, in every perfect Bayesian equilibrium and for every disclosure policy and temporal structure in the class. Consequently $Y_t$ and each of its conditional expectations are determined by the primitives, temporal structure, and disclosure policy; under no and full prior-outcome disclosure the effective prize spreads are given by~\eqref{eq:VtN} and~\eqref{eq:VtF}. Given these spreads, play in each battle is pinned down: the unique component equilibrium of Lemma~\ref{lem:linear-reduction} when the spread is positive, and zero expenditure when it is zero. The distribution of output in every battle, and hence expected total output, is therefore the same in every perfect Bayesian equilibrium.

\section{Disclosure and Temporal-Refinement Proofs}\label{app:disclosure-refinement-proofs}

\subsection{Proof of Corollary~\ref{cor:garbling}}\label{app:garbling}

Fix $\mathcal P$ and a committed, public, and common garbling with public information $\mathcal G_\ell=\sigma(S_1,\ldots,S_\ell)$. Because the winning probability in battle $t$ is $\mu_t$ at every public history and every positive spread, and the zero-spread convention preserves the same recorded probability, induction over blocks implies that every admissible garbling generates the same product distribution of battle outcomes. The exogenous signal randomizations are independent of this outcome vector.

Let $\boldsymbol{\xi}_{\ell}=(\xi_1,\ldots,\xi_\ell)$. Every signal observed
before block $\ell$ is measurable with respect to
$\sigma(\mathcal H_\ell^{\mathcal P},\boldsymbol{\xi}_\ell)$, and hence
$\mathcal G_\ell\subseteq
\sigma(\mathcal H_\ell^{\mathcal P},\boldsymbol{\xi}_\ell)$. Independence of
the designer's randomization devices gives
\[
\E\left[Y_t\mid
\mathcal H_{\ell(t)}^{\mathcal P},\boldsymbol{\xi}_{\ell(t)}\right]
=
V_t^{F,\mathcal P}\left(\mathcal H_{\ell(t)}^{\mathcal P}\right).
\]
The tower property therefore gives
\[
V_t^S
=
\E\left[
V_t^{F,\mathcal P}\left(\mathcal H_{\ell(t)}^{\mathcal P}\right)
\mid\mathcal G_{\ell(t)}
\right],
\qquad
V_t^N=\E[V_t^S].
\]
Thus $V_t^S$ lies between no disclosure and full disclosure in the conditional-expectation order. If $\phi_t$ is concave, two applications of Jensen's inequality yield
\[
\phi_t(V_t^N)
\geq
\E[\phi_t(V_t^S)]
\geq
\E\left[\phi_t\left(V_t^{F,\mathcal P}\left(\mathcal H_{\ell(t)}^{\mathcal P}\right)\right)\right].
\]
Summing over battles proves part 1. Convexity reverses both inequalities, and linearity gives equality.

\subsection{Proof of Proposition~\ref{prop:expenditure-neutrality}}\label{app:expenditure}

In each component battle, equilibrium ability-scaled expenditure scales proportionally with the effective prize spread. In particular, under the finite-support specification the equilibrium distribution of $y/v$ is independent of $v$ (see the proof of Lemma~\ref{lem:linear-reduction}). Because expenditure satisfies $e=y/\alpha$, expenditure scales proportionally with $v$ for every type. Moreover, battle-$t$ types are drawn independently of the public history that determines the effective prize spread. Hence, conditional on any spread $v>0$, expected aggregate expenditure in battle $t$ is $c_tv$, where by the second expression in~\eqref{eq:linear-expenditure}
\[
c_t=\sum_{k=1}^{m_t}\frac{p_k^{A(t)}}{\alpha_k^{A(t)}}\left(I_k^t+I_{k+1}^t\right)\geq0
\]
depends only on the primitives of battle $t$; at $v=0$, equilibrium expenditure is zero.

Let $V_t$ denote the effective prize spread induced by the design. Because expected expenditure in battle $t$ is linear in $V_t$, the law of iterated expectations gives
\[
\E[c_tV_t]=c_t\E[Y_t].
\]
Changing the disclosure policy or temporal structure may therefore raise $V_t$ after some histories and lower it after others, but neither kind of change affects its unconditional mean, and the changes exactly offset in expectation. By Proposition~\ref{prop:pbe-assembly}, the distribution of the outcome vector, and hence
\[
\E[Y_t]=\pi_t+V\Theta_A(n;\mathcal T\setminus\{t\}),
\]
is the same for every disclosure policy in the class and every ordered temporal structure. Summing over battles gives the stated value, which depends only on the primitives of the component battles.

Unlike Theorem~\ref{thm:information-refinement}, this argument uses only
unconditional means and therefore applies to nonnested temporal structures
as well. It establishes no convex-order comparison between them.

\section{Continuous-Type Extension}\label{app:continuous-type-extension}

\subsection{Proof of Proposition~\ref{prop:continuous-types}}\label{app:continuous-types}

Fix battle $t$ and suppress its subscript. Differentiating the defining equation for $G$ gives
\[
G'(\theta)
=
\frac{\alpha^B(G(\theta))}{\alpha^A(\theta)},
\]
so
\[
K'(\theta)
=
\alpha^A(\theta)G'(\theta)
=
\alpha^B(G(\theta)).
\]
For $v>0$, consider the strategies
\[
y^A(\theta)=vK(\theta),
\qquad
y^B(G(\theta))=vK(\theta),
\]
with types $\eta<\underline\theta$ of player $B$ choosing zero.

By the same strategic transformation used in
Lemma~\ref{lem:linear-reduction}, a player of ability $\alpha$ pays
$y/\alpha$. If type $\theta$ of player $A$ chooses $vK(r)$, it wins with
probability $G(r)$ and obtains
\[
U^A(r\mid\theta)
=
vG(r)-\frac{vK(r)}{\alpha^A(\theta)}.
\]
Hence
\[
\frac{\partial U^A(r\mid\theta)}{\partial r}
=
vG'(r)
\left(1-\frac{\alpha^A(r)}{\alpha^A(\theta)}\right),
\]
which is positive for $r<\theta$ and negative for $r>\theta$. Similarly, an active type $G(\theta)$ of player $B$ choosing $vK(r)$ obtains
\[
U^B(r\mid\theta)
=
vr-\frac{vK(r)}{\alpha^B(G(\theta))},
\]
with derivative
\[
\frac{\partial U^B(r\mid\theta)}{\partial r}
=
v\left(1-
\frac{\alpha^B(G(r))}{\alpha^B(G(\theta))}
\right),
\]
so $r=\theta$ is optimal. For $\eta<\underline\theta$, the payoff from $vK(r)$ has derivative
\[
v\left(1-
\frac{\alpha^B(G(r))}{\alpha^B(\eta)}
\right)<0,
\]
so it is strictly decreasing in $r$. Since this payoff converges to zero
as $r\downarrow0$ and choosing $y=0$ yields zero, zero is optimal. No type benefits from choosing above the common support. This establishes the proposed monotone equilibrium.

For uniqueness, \citet{amann-leininger-1996} establish that any
equilibrium in this atomless two-player environment has monotone
representatives and the standard common-support properties. Thus, let
$(\tilde y^A,\tilde y^B)$ represent any equilibrium, with induced
distributions $H^A$ and $H^B$ of ability-scaled expenditure. The two
distributions have a common support $[0,\bar y]$ without gaps, neither
has an atom at any $y>0$, and at most one has an atom at zero.
Monotonicity therefore defines a strictly increasing map $\tilde G$
matching the type of player $A$ choosing $y$ to the type of player $B$
choosing the same $y$, with $\tilde G(1)=1$ because the two highest
types choose $\bar y$. Type $\theta$ of player $A$ then wins with
probability $\tilde G(\theta)$, and type $\tilde G(\theta)$ of player
$B$ wins with probability $\theta$, so optimality of the assigned
choices requires
\[
(\tilde y^A)'(\theta)=v\alpha^A(\theta)\tilde G'(\theta)
\qquad\text{and}\qquad
(\tilde y^A)'(\theta)=v\alpha^B(\tilde G(\theta))
\]
at almost every $\theta$. Hence $\alpha^A(\theta)\tilde G'(\theta)=\alpha^B(\tilde G(\theta))$ with $\tilde G(1)=1$, which is~\eqref{eq:continuous-matching}. Because $\alpha^B$ is continuously differentiable and $\alpha^A$ is bounded away from zero, the right-hand side of this differential equation is Lipschitz in $\tilde G$, so the terminal value problem has a unique solution. Therefore $\tilde G=G$ and, integrating from $\tilde y^A(0)=0$, $\tilde y^A=vK$ almost everywhere.

Player $A$'s expected output is
\[
\int_0^1\gamma(vK(\theta))\,d\theta.
\]
For player $B$, changing variables $\eta=G(\theta)$ gives
\[
\int_{\underline\theta}^1
\gamma\!\left(vK(G^{-1}(\eta))\right)d\eta
=
\int_0^1\gamma(vK(\theta))G'(\theta)\,d\theta.
\]
Adding the two expressions proves~\eqref{eq:continuous-phi}. Moreover,
\[
\mu^{\mathrm{ct}}
=
\int_0^1G(\theta)\,d\theta,
\]
which is independent of $v>0$ and, since $G$ is strictly increasing with $G(1)=1$, lies in $(0,1)$. Because~\eqref{eq:continuous-phi} is a positive weighted integral of $v\mapsto\gamma(vK(\theta))$, $\phi^{\mathrm{ct}}$ has the same curvature as $\gamma$. At $v=0$, zero expenditure is uniquely optimal, so $\phi^{\mathrm{ct}}(0)=0$.

\subsection{Proof of Corollary~\ref{cor:continuous-type-extension}}\label{app:cor:continuous-type-extension}

Proposition~\ref{prop:continuous-types} supplies exactly the component-battle properties used in the finite-support analysis: $\phi_t^{\mathrm{ct}}(0)=0$, the positive-spread winning probability $\mu_t^{\mathrm{ct}}$ is level-invariant, the output response has the same curvature as $\gamma$, and ability-scaled expenditure is proportional to the effective prize spread. The last property supplies the same common normalized-expenditure draw used in Theorem~\ref{thm:information-refinement}, so its player-level martingale coupling and convex-order conclusion continue to hold. With independent type draws and the maintained zero-spread recording convention, the proofs of Proposition~\ref{prop:pbe-assembly}, Proposition~\ref{prop:disclosure-ranking}, Corollary~\ref{cor:garbling}, and Proposition~\ref{prop:temporal-refinement} therefore apply after replacing $(\phi_t,\mu_t)$ by $(\phi_t^{\mathrm{ct}},\mu_t^{\mathrm{ct}})$. The conclusion of Corollary~\ref{cor:unrestricted-clustering} follows from the same argument.

For expenditure neutrality, proportionality implies that conditional expected aggregate expenditure in battle $t$ is $c_t^{\mathrm{ct}}v$ for some finite constant $c_t^{\mathrm{ct}}\geq0$ depending only on that battle's primitives. Hence, for the induced spread $V_t$,
\[
\E[c_t^{\mathrm{ct}}V_t]
=
c_t^{\mathrm{ct}}\E[Y_t]
\]
by iterated expectations. Since the outcome distribution, and therefore $\E[Y_t]$, is invariant across the stated disclosure policies and temporal structures, expected aggregate expenditure is invariant as well.

\section{Type-Observability Proofs}\label{app:observable-types}

We first characterize a component battle conditional on realized abilities, and then apply the same conditional-expectation argument under local type observability and global type revelation.

\subsection{Observable opponent types within component battles}\label{subsec:complete-one-battle}

Fix battle $t$ and let
\[
\boldsymbol\alpha_t=(\alpha^{A(t)},\alpha^{B(t)}),
\qquad
\underline\alpha_t=\min\{\alpha^{A(t)},\alpha^{B(t)}\},
\qquad
\overline\alpha_t=\max\{\alpha^{A(t)},\alpha^{B(t)}\}.
\]
Conditional on $\boldsymbol\alpha_t$, the battle is the two-player all-pay auction with common cost function $\gamma^{-1}$ and values $\alpha^{A(t)}v$ and $\alpha^{B(t)}v$, whose equilibrium is characterized by \citet{kaplan-wettstein-2006}.

\begin{lemma}[Complete-information battle]\label{lem:complete-allpay}
For $v>0$, conditional on $\boldsymbol\alpha_t$, expected total output in the complete-information component battle is
\begin{equation}
\psi_t(v\mid\boldsymbol\alpha_t)
=
\underline\alpha_t
\left(\frac{1}{\alpha^{A(t)}}+\frac{1}{\alpha^{B(t)}}\right)
\int_0^1\gamma(\underline\alpha_t vz)\,dz,
\label{eq:psi-realized-change}
\end{equation}
and team $A$ wins with probability
\[
q_t^C(\boldsymbol\alpha_t)
=
\begin{cases}
1-\dfrac{\alpha^{B(t)}}{2\alpha^{A(t)}},
& \alpha^{A(t)}\geq\alpha^{B(t)},\\[6pt]
\dfrac{\alpha^{A(t)}}{2\alpha^{B(t)}},
& \alpha^{A(t)}<\alpha^{B(t)}.
\end{cases}
\]
In either case $q_t^C\in(0,1)$, and the winning probability is independent of $v$. At $v=0$, zero expenditure is the unique equilibrium and $\psi_t(0\mid\boldsymbol\alpha_t)=0$.
\end{lemma}

\begin{proof}
Assume without loss of generality that $\alpha^{A(t)}\geq\alpha^{B(t)}$ and suppress $t$. By Lemma~\ref{lem:linear-reduction}, the battle is strategically equivalent to the linear-cost all-pay auction with values $\alpha^Av\geq\alpha^Bv$, whose standard equilibrium \citep{baye-kovenock-de-vries-1996} has common support $[0,\alpha^Bv]$, with
\[
F^A(y)=\frac{y}{\alpha^Bv},
\qquad
F^B(y)=1-\frac{\alpha^B}{\alpha^A}+\frac{y}{\alpha^Av}.
\]
Using $x=\gamma(y)$, expected total output is
\[
\frac1v\left(\frac1{\alpha^A}+\frac1{\alpha^B}\right)
\int_0^{\alpha^Bv}\gamma(y)\,dy,
\]
which becomes~\eqref{eq:psi-realized-change} after the change of variables $y=\alpha^Bvz$. The winning probability of player $A$ is
\[
\int_0^{\alpha^Bv}F^B(y)\,dF^A(y)
=
1-\frac{\alpha^B}{2\alpha^A}.
\]
The opposite ability ranking is symmetric. The statement for $v=0$ is immediate.
\end{proof}

Under local type observability, define
\[
\phi_t^C(v)=\E_{\boldsymbol\alpha_t}[\psi_t(v\mid\boldsymbol\alpha_t)],
\qquad
\mu_t^C=\E_{\boldsymbol\alpha_t}[q_t^C(\boldsymbol\alpha_t)].
\]
Equation~\eqref{eq:psi-realized-change} implies that $\phi_t^C$ has the same curvature as $\gamma$, while $\mu_t^C$ is independent of every positive spread. Replacing $\mu_s$ by $\mu_s^C$ in the pivotality probabilities defines the corresponding spreads $V_t^{C,N}$ and $V_t^{C,F,\mathcal P}$. Iterated expectations give
\[
V_t^{C,N}
=
\E\left[V_t^{C,F,\mathcal P}
\left(\mathcal H_{\ell(t)}^{\mathcal P}\right)\right].
\]
Jensen's inequality therefore yields the fixed-structure disclosure ranking in Corollary~\ref{cor:type-observability} for local type observability. The restriction $\pi_t>0$ ensures that all relevant spreads are positive.

\subsection{Global revelation of type realizations}\label{subsec:global-type-app}

Let $\boldsymbol\alpha=(\boldsymbol\alpha_s)_{s\in\mathcal T}$ denote the publicly observed type profile. Conditional on $\boldsymbol\alpha$, battle outcomes are independent with winning probabilities $q_s^C(\boldsymbol\alpha_s)$. Define
\[
V_t^{G,N}(\boldsymbol\alpha)
=
\E[Y_t\mid\boldsymbol\alpha],
\qquad
V_t^{G,F,\mathcal P}
(\boldsymbol\alpha,\mathcal H_{\ell(t)}^{\mathcal P})
=
\E[Y_t\mid\boldsymbol\alpha,
\mathcal H_{\ell(t)}^{\mathcal P}].
\]
The tower property gives
\[
V_t^{G,N}(\boldsymbol\alpha)
=
\E\left[
V_t^{G,F,\mathcal P}
(\boldsymbol\alpha,\mathcal H_{\ell(t)}^{\mathcal P})
\mid\boldsymbol\alpha
\right].
\]
Conditional expected output is evaluated by $\psi_t(\cdot\mid\boldsymbol\alpha_t)$, which has the same curvature as $\gamma$. Conditional Jensen's inequality therefore gives the fixed-structure disclosure ranking in Corollary~\ref{cor:type-observability} for global type revelation.

\subsection{Expenditure under alternative type observability}

The complete-information equilibrium used above implies that, conditional on $\boldsymbol\alpha_t$ and $v>0$,
\[
\E[e^{A(t)}+e^{B(t)}\mid\boldsymbol\alpha_t,v]
=
\kappa_t(\boldsymbol\alpha_t)v,
\qquad
\kappa_t(\boldsymbol\alpha_t)
=
\frac{\underline\alpha_t}{\overline\alpha_t}.
\]
In particular $\kappa_t$ depends on the realized ability pair alone and not on $\gamma$, which is the component-level form of the cost-function invariance of expected expenditure established by \citet{kaplan-wettstein-2006}. Under local type observability, $\kappa_t(\boldsymbol\alpha_t)$ is independent of the history-generated effective prize spread, so
\[
\E\left[
\kappa_t(\boldsymbol\alpha_t)
V_t^{C,F,\mathcal P}
\right]
=
\E[\kappa_t(\boldsymbol\alpha_t)]V_t^{C,N}.
\]
Under global type revelation, $\kappa_t(\boldsymbol\alpha_t)$ is measurable with respect to $\boldsymbol\alpha$, and the preceding conditional-expectation identity gives
\[
\E\left[
\kappa_t(\boldsymbol\alpha_t)
V_t^{G,F,\mathcal P}
\right]
=
\E\left[
\kappa_t(\boldsymbol\alpha_t)
V_t^{G,N}(\boldsymbol\alpha)
\right].
\]
Thus expected expenditure in every battle is the same under no and full prior-outcome disclosure in either observability regime. Summing over battles completes the proof of Corollary~\ref{cor:type-observability}.

\begingroup
\renewcommand{\baselinestretch}{1.24}\normalsize
\setlength{\bibsep}{0pt}
\widowpenalty=10000
\clubpenalty=10000

\endgroup

\end{document}